\documentclass[journal]{IEEEtran}
\usepackage{setspace}

\usepackage{lipsum}
\usepackage{makeidx}
\usepackage{enumerate}
\usepackage{color}
\usepackage{cite}
\usepackage{amsmath,amsthm}   
\usepackage{amssymb}
\usepackage{nomencl}
\usepackage{multirow}
\usepackage{graphicx}
\usepackage{multirow}
\usepackage{anysize}
\usepackage{float}
\usepackage{epstopdf}
\usepackage{threeparttable}
\usepackage{multicol}
\DeclareGraphicsExtensions{.pdf,.jpeg,.png,.jpg,.emf,.eps}
\usepackage{geometry}

\usepackage{calc}
\usepackage{here}
\usepackage{url}
\usepackage{upref}
\usepackage{comment}
\usepackage{times}
\usepackage{dsfont}
\usepackage{epic,eepic}
\usepackage{rawfonts}
\usepackage[T1]{fontenc}
\usepackage{latexsym}
\usepackage{amsfonts}

\usepackage{capitalgreekitalic}

\usepackage{geometry}
\usepackage[font=small,labelfont=bf]{caption}
\usepackage[font=small,labelfont=bf]{subcaption}

\usepackage{tikz,pgfplots}
\usetikzlibrary{calc}
\usetikzlibrary{intersections}

\newtheorem{theorem}{Theorem}

\newtheorem{lemma}{Lemma}

\newtheorem{remark}{Remark}

\theoremstyle{definition}
\newtheorem{definition}{Definition}

\begin{document}

\title{Robust Improper Signaling for Two-user SISO Interference Channels}

\author{Mohammad Soleymani$^*$, \emph{Student Member, IEEE}, Christian Lameiro$^*$, \emph{Member, IEEE}, Ignacio Santamaria$^\dag$ \emph{Senior Member, IEEE} and Peter J. Schreier$^*$, \emph{Senior Member, IEEE}
\\ \thanks{ 
$^*$Mohammad Soleymani, Christian Lameiro and Peter J. Schreier are with the Signal and System Theory Group, Universit\"at Paderborn, Germany, http://sst.upb.de  (emails: \protect\url{{mohammad.soleymani, christian.lameiro, peter.schreier}@sst.upb.de}).  

$^\dag$Ignacio Santamaria is with the Department of Communications Engineering, University of Cantabria (email: \protect\url{i.santamaria@unican.es}).

The work of M. Soleymani, C. Lameiro and P. J. Schreier was supported by the German Research Foundation
(DFG) under grants LA 4107/1-1 and SCHR 1384/8-1. 

The work of I. Santamaria was supported by MINECO of Spain and AEI/FEDER funds of the E.U., under grant TEC2016-75067-C4-4-R (CARMEN).
}

}

\maketitle

\begin{abstract}
It has been shown that improper Gaussian signaling (IGS) can improve the performance of wireless interference-limited systems when perfect channel state information (CSI) is available. In this paper, we investigate the robustness of IGS against imperfect CSI on the transmitter side in a two-user single-input single-output (SISO) interference channel (IC) as well as in a SISO Z-IC, when interference is treated as noise. We assume that the true channel coefficients belong to a known region around the channel estimates, which we call the uncertainty region. Following a worst-case robustness  approach, 
we study the rate-region boundary of the IC for the worst channel in the uncertainty region. For the two-user IC, we derive a robust design in closed-form, which is independent of the phase of the channels by allowing only one of the users to transmit IGS. For the Z-IC, we provide a closed-form design for the transmission parameters by considering an enlarged uncertainty region  and allowing both users to employ IGS. In both cases, the IGS-based designs are ensured to perform no worse than proper Gaussian signaling. Furthermore, we show, through numerical examples, that the proposed robust designs significantly outperform non-robust solutions. 
\end{abstract} 

\section{Introduction}
The increasing number of devices with wireless connectivity, on the one hand, and the limited availability of radio resources, on the other hand, motivate the design of techniques to exploit the spectrum more efficiently. 
As a consequence, modern wireless communications systems are mostly {\em interference-limited}. 
Hence, interference management techniques play an important role in improving the performance of such systems. 
A way to increase the spectral efficiency of interference-limited systems is to employ improper Gaussian signaling (IGS). 
In IGS, the real and imaginary parts of the signal are correlated and/or have unequal powers \cite{schreier2010statistical}. 
Although proper Gaussian signaling (PGS) achieves the capacity of traditional wireless communication channels such as point-to-point communications \cite{cover2012elements}, IGS has been shown to improve the performance of several interference-limited systems, such as interference channels (IC) \cite{cadambe2010interference,yang2014interference, ho2012improper, zeng2013transmit, nguyen2015improper, lameiro2013degrees, lagen2016superiority, kurniawan2015improper,  lameiro2017rate}, underlay and overlay cognitive radio (CR) systems \cite{lameiro2015benefits,amin2016underlay,amin2017overlay}, relay networks \cite{ho2013optimal, javed2017full, gaafar2018full}, and broadcast channels (BC) with (widely) linear transceivers \cite{hellings2013performance,hellings2013qos}, to mention a few.

\subsection{Related work}
IGS was studied as an effective interference management tool for the first time in \cite{cadambe2010interference},  where the authors showed that IGS can increase the degrees of freedom (DoF) in the three-user IC. 
In \cite{yang2014interference}, the authors showed that IGS increases the DoF of a multiple-input multiple-output (MIMO) two-user X-IC\footnote{The two-user X-IC is a generalization
of the two-user IC where there is an independent message from each transmitter to each receiver \cite{huang2012interference}.}. 
In \cite{ho2012improper}, the authors showed that IGS can improve  the performance of a single-input single-output (SISO) two-user IC. 
In \cite{zeng2013transmit}, the authors studied the achievable rate of IGS in MIMO ICs and proposed algorithms to derive the rate region of the two-user SISO IC. 
The paper \cite{zeng2013transmit} showed that IGS can enlarge the rate region of the two-user SISO-IC.    
In \cite{nguyen2015improper}, the authors showed that IGS can reduce the symbol error rate in a $K$-user IC. 
 In \cite{lagen2016superiority, kurniawan2015improper, lameiro2017rate}, it was shown that IGS can increase the achievable rate of the Z-IC\footnote{Z-IC  is a special case of the two-user IC, in which only one of the users interferes with the other user.} in different scenarios. 
 The authors in \cite{lameiro2015benefits} considered an underlay CR (UCR) system and derived a condition on the ratio between the gain of the interference and direct channel coefficients for IGS to outperform PGS.

All the aforementioned works assumed that {\em global} channel state information (CSI) is 
available at every transmitter and receiver, meaning that the channel
knowledge at the transmitters is 
instantaneous and globally available. 
However, global and perfect CSI at every transmitter is a very restrictive assumption. 
Therefore, a great deal of 
work has aimed at relaxing this assumption and exploiting either imperfect or statistical CSI at the transmitter (CSIT). 
For instance, \cite{amin2016underlay} studied an underlay CR system, in which the secondary user (SU)  transmitter has access only to the average CSI. This investigation showed that IGS achieves lower outage probability than its proper counterpart.
In \cite{amin2017overlay}, the authors considered both global CSI and partial CSI scenarios in an overlay CR system. 
Specifically, with partial CSIT the primary
 user has access only to 
 the average CSI. 
In this setting, the authors showed that IGS reduces the outage probability of the PU link. 

In the aforementioned papers, it was assumed that the existing CSI, either instantaneous or statistical, is perfect. 
However, in  practical scenarios, the CSI is always imperfect. Furthermore, acquiring the CSI at the transmitter side is more difficult than at the receiver side. 
Thus, it is critical to investigate whether IGS is still beneficial when the available CSI is imperfect. 
To design the PGS scheme, only the channel gains are required. 
However, IGS typically requires the phases of the channels in addition to the gains in order to design the optimal complementary variances. 
Since IGS needs more detailed CSI than PGS, it might be reasonable to expect  that IGS is more affected by imperfect CSI and its benefit will decrease when the CSI quality decreases as well.  

On the other hand, the use of IGS in combination with interference alignment (IA)  obtains DoF improvements in certain scenarios such as three-user IC \cite{cadambe2010interference} or two-user MIMO X-IC \cite{yang2014interference}. 
In other scenarios, such as the two-user IC with partial or full transmitter cooperation, the use of IGS does not offer DoF advantages, although IA  still offers improvements in terms of generalized degrees of freedom (GDoFs), which are a refinement of the DoF metric. 
However, these asymptotic (in the high SNR regime) DoF or GDoF benefits are typically lost under finite precision CSIT, a phenomenon known as DoF collapse. 
For instance, under limited CSIT, the DoF collapse for the two-user multiple-input, single-output (MISO) broadcast channel (BC) even when perfect channel knowledge for one user is available \cite{tandon2013synergistic,davoodi2016aligned}. 
It is also shown in  \cite{davoodi2017transmitter}  that, under finite precision CSIT, the sum GDoF of the two-user X channel and two-user BC collapse and hence, the benefits of IA are entirely lost under finite precision CSIT.

In this paper, we do not study possible DoF or GDoF benefits of IGS, but rather the advantages that IGS can provide in terms of rate when the interference is treated as noise.  
Treating interference as noise (TIN) is an attractive technique because of its simplicity, which also turns out to be optimal in terms of GDoF if the desired signals at every receiver are strong enough \cite{geng2015optimality}. 
Nevertheless, the optimal IGS parameters typically depend on the gains and phases of the channels, so the following question arises:
{\em 
Is IGS still beneficial in the presence of imperfect CSI?} 
We will address this question in this work.

\subsection{Contribution}
In this paper, 
we investigate the robustness of IGS against imperfect CSI  on the transmitter side. 
To the best of our knowledge, this is the first work to analyze the effect of imperfect CSI on IGS. 
We employ a 
 {\em worst-case} robustness approach to derive optimal robust IGS designs \cite{ben2009robust,vorobyov2003robust,tajer2011robust,wang2013robust, feng2015robust,   mochaourab2012robust, shaverdian2014robust, wang2015robust, pascual2006robust}. Thereby, we assume that the true channels are within a known bounded region around the available CSI estimate with a certain probability, which we call the uncertainty region. Robustness is then achieved by performing an optimization for the worst-case channels within the uncertainty region. 
 
 We study the robustness issue in the two-user IC and in the Z-IC when the CSI at the transmitter side is imperfect.
We consider a scenario in which the transmitters have access to a noisy estimate of the channel coefficients, 
and we propose closed-form robust designs of the transmission parameters that achieve the robust rate region of the two-user IC and Z-IC with TIN. 
To this end, we first extend the results in \cite{lameiro2015benefits} to derive a worst-case robust design for the two-user IC. 
In \cite{lameiro2015benefits}, an UCR scenario was considered, in which the primary user (PU), unaware of the secondary user (SU), employs PGS and transmits with maximum power. 
In UCR, there are two types of users, PUs and SUs. 
PUs are licensed users and have priorities to use the resources; on the other hand, SUs can transmit only if they do not disturb the communications of the PUs.
However, in the two-user IC, a higher degree of cooperation between users may be allowed to achieve a better performance. 
Thus, we  derive the Pareto-optimal boundary of the robust rate region for the two-user IC when at most one user employs IGS, which makes the scheme robust against imperfect CSI and represents a generalization of the scenario studied in \cite{lameiro2015benefits}. 
We then extend the design in  \cite{lameiro2017rate} to a robust design for the Z-IC. 
The main challenge of the robust design for the Z-IC is to consider the phase uncertainty, which makes the optimization problems more difficult than in the setting of \cite{lameiro2017rate}. 
  The main contributions of this work are  the following:
\begin{itemize}

\item We first consider a two-user IC scenario and derive a robust scheme in closed-form 
by  allowing only one of the users to employ IGS. 
 An important advantage of this suboptimal scheme is that it is independent of the channel phase information at the transmitters. We derive a sufficient and necessary condition for the proposed IGS scheme to outperform PGS in this scenario.

\item 
Although the results for the two-user IC can be applied to the Z-IC, we propose another robust design for the Z-IC that allows both users to employ IGS simultaneously.   
We also 
derive closed-form conditions for this IGS scheme to outperform PGS. 
These conditions, which depend on the accuracy of the CSI at the transmitter side, provide a robust version of the design in \cite{lameiro2017rate}. 

\item 
Our results show 
that even in the presence of imperfect CSI, robust IGS can substantially outperform robust PGS designs. This is in contrast to non-robust IGS designs, which may be strongly affected by uncertain CSI. 
Our results show that improper signaling is even more robust to imperfect CSI than its proper counterpart when the transmission parameters are optimized in a robust way.
\end{itemize}

\subsection{Paper outline}
The rest of this paper is organized as follows. In Section \ref{secII}, we describe the system model and assumptions the assumptions made about the CSI knowledge and formulate the optimization problem. 
  In Section \ref{IX-ch}, we derive the robust rate region for the two-user IC.  
We derive the robust design of the parameters for the Z-IC in Section \ref{secIII}.
In Section \ref{secIV}, we present some numerical results.

\section{System model and problem statement}\label{secII}
In this section, we describe our system model and define all parameters. 
Section \ref{secII-A} briefly presents some background material on improper signal processing used in the paper. 
In Section \ref{secII-B}, we describe the system model and Section \ref{secII-C} presents the imperfect CSI model. 
Section \ref{secII-D} describes the optimization problem to derive the boundary of the robust rate region.
\subsection{Preliminaries of IGS}\label{secII-A}
A zero-mean complex Gaussian random variable $x$ can be uniquely specified by its variance, $p_x=\mathbb{E}\{|x|^2\}$, and complementary variance, $\mathbb{E}\{x^2\}$ \cite{schreier2010statistical}.  
The circularity coefficient of a complex Gaussian random variable is defined as $\kappa_x=\frac{|\mathbb{E}\{x^2\}|}{\mathbb{E}\{|x|^2\}}$, which takes values between 0 and 1. 
We call $x$ proper if $\kappa_x=0$; otherwise, we call it improper. Moreover, we call $x$ maximally improper if $\kappa_x=1$.
The complementary variance of $x$ can be written as $\mathbb{E}\{x^2\}=p_x\kappa_x e^{j\phi_x}$, where $\phi_x$ is the phase of $\mathbb{E}\{x^2\}$.

\subsection{Signal model}\label{secII-B}
 \begin{figure}[t]
\centering
\includegraphics[width=0.32\textwidth]{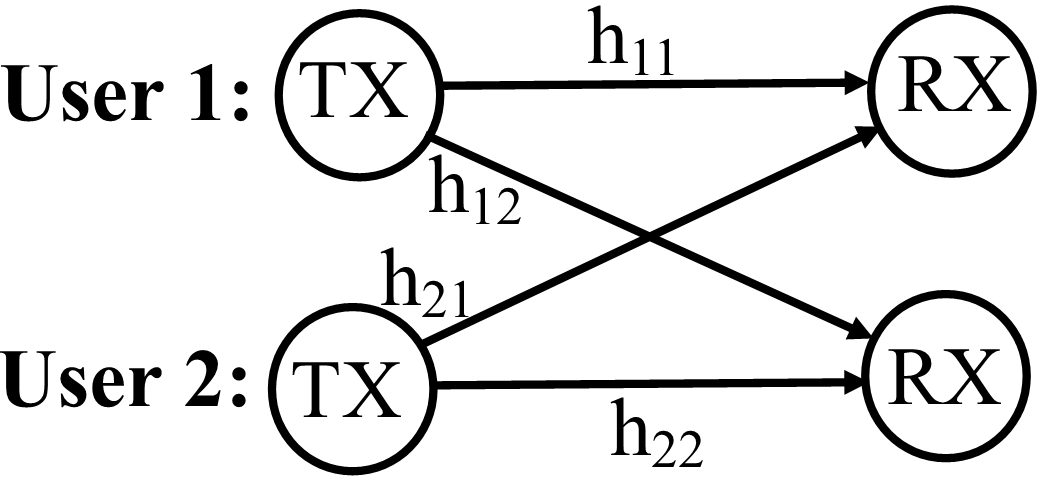}
\caption{Channel model for the two-user SISO IC.}
\label{Fig1}
\end{figure}
We consider a two-user IC, as depicted in Fig. \ref{Fig1}, where both users may transmit improper Gaussian signals. 
In the two-user IC, transmitter $i$ transmits its message, $x_i$, to receiver $i$.
The received signals for user $i\in\{1,2\}$ is\footnote{Note that in the Z-IC, the coefficient $h_{12}=0$.} 
\begin{equation}
y_i=h_{ii}x_i+h_{\bar{\imath}i}x_{\bar{\imath}}+n_i,
\end{equation}
where $\bar{\imath}\in\{1,2\}$ and $\bar{\imath}\neq i$. 
Moreover, $x_i$, $h_{ji}$, $n_i$ for $i,j\in\{1,2\}$ are, respectively, the transmit signal of user $i$, the link from transmitter $j$ to receiver $i$, and the additive noise at the $i$th receiver, which is assumed to be zero-mean proper complex Gaussian with variance $\sigma^2$.
The transmitted signals $x_1$ and $x_2$ are complex Gaussian and may be improper.

Since we treat interference as noise,  the achievable rate of user $i\in\{1,2\}$ is \cite{lameiro2017rate,zeng2013transmit}
\begin{multline}
R_i=  
\frac{1}{2}\!\log_2\!\!\left(\!\!\frac{\left(\sigma^2\!+\sum_{j=1}^{2}p_j|h_{ji}|^2\!\right)^2\!\!\!}{(p_{\bar{\imath}}|h_{\bar{\imath}i}|^2+\sigma^2)^2-\left(p_{\bar{\imath}}\kappa_{\bar{\imath}}|h_{\bar{\imath}i}|^2\right)^2}\right.\\
\left.-\frac{\!\left|\sum_{j=1}^{2} p_j\kappa_j|h_{ji}|^2e^{j(2\measuredangle h_{ji}+\phi_j)}\right|^2}
{(p_{\bar{\imath}}|h_{\bar{\imath}i}|^2+\sigma^2)^2-\left(p_{\bar{\imath}}\kappa_{\bar{\imath}}|h_{\bar{\imath}i}|^2\right)^2}\!\right)\!\!,
\label{R1}
\end{multline}
where $p_j$, $\kappa_j$, and $\phi_j$ for $j=1,2,$ are, respectively, the transmission power, circularity coefficient and phase of the complementary covariance of the signal transmitted by user $j$. 
Moreover,  $|h_{ij}|$ and $\measuredangle h_{ij}$ for $i,j\in\{1,2\}$ are the magnitude and phase of the channel from the $i$th transmitter to the $j$th receiver, respectively. 
Note that the rate expressions in \eqref{R1} are the maximum rates that can be supported by the channels. 
\subsection{Uncertainty model}\label{secII-C}
We assume perfect CSI at the receivers but imperfect CSI at the transmitters. 
 It is reasonable to assume that a receiver knows the CSI perfectly since acquiring CSI at the receiver side is relatively easy with training sequences or applying  blind/semi-blind estimation methods \cite{shen2013robust, wang2009worst, pei2011secure}. 
 On the other hand, the channel information is typically quantized and then sent to the transmitters through a noisy feedback link. Therefore we have imperfect CSIT \cite{shen2013robust, wang2009worst, shaverdian2014robust, pei2011secure,wei2017optimal}.
 In this paper, 
 we assume that the true channel, $h_{ij}$ for $i,j\in\{1,2\}$, lies in a vicinity of the channel estimate at the transmitter side, $\hat{h}_{ij}$, i.e., $\hat{h}_{ij}=h_{ij}+e_{ij}$, where $e_{ij}$ accounts for all sources of error between 
 the estimate and the true channel. 
 We assume that the true channel $h_{ij}$ belongs to an uncertainty set $\mathcal{E}_{ij}$, which includes $\hat{h}_{ij}$. 
 We do not restrict our model to any specific error source and consider an arbitrary model for the uncertainty set, as illustrated in Fig. \ref{Fig2-1}. 
 \begin{figure}[t]
\centering
\includegraphics[width=0.36\textwidth]{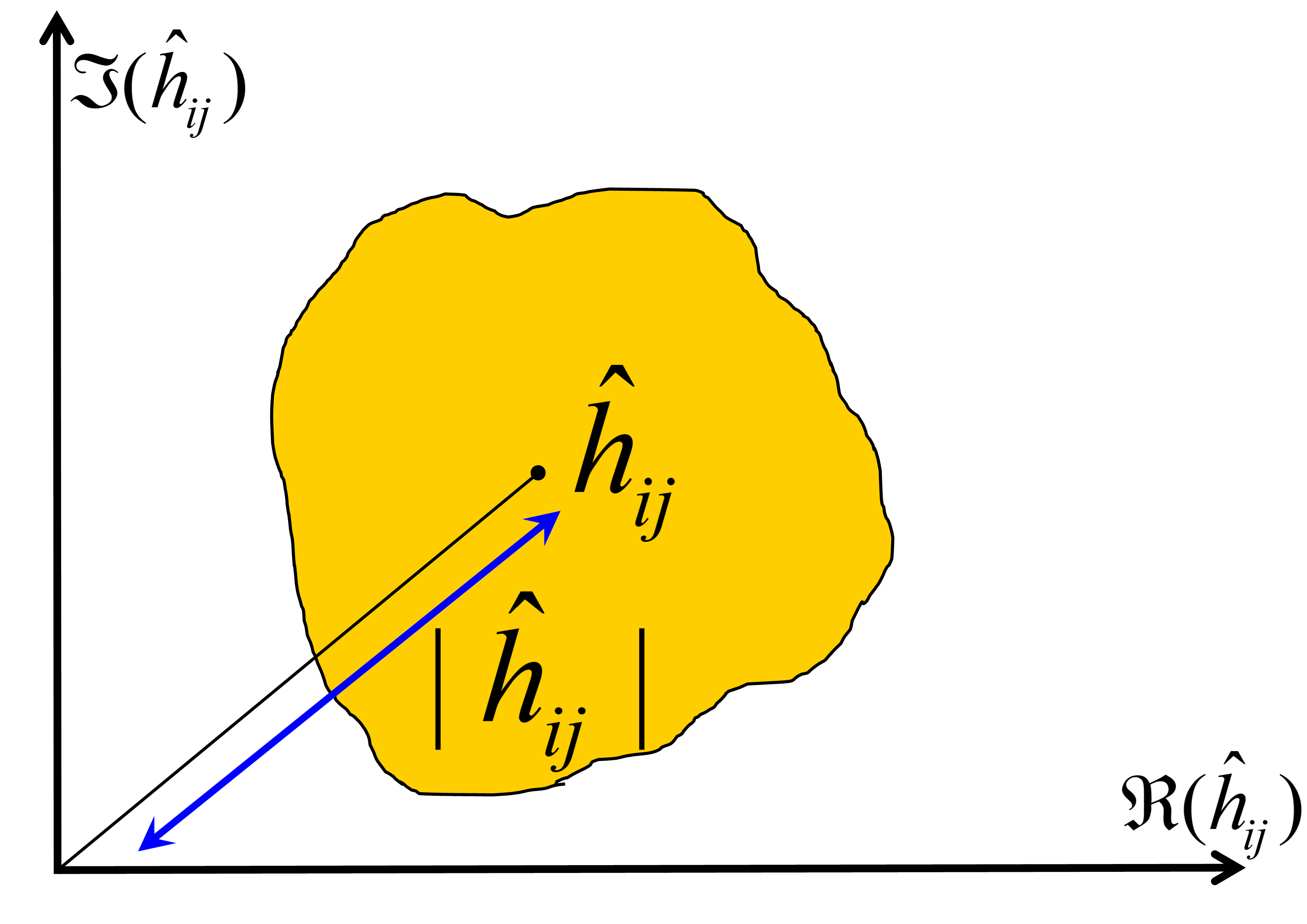}
\caption{Arbitrary channel uncertainty region.}
\label{Fig2-1}
\end{figure}

\subsection{Problem statement}\label{secII-D}
Our aim in this work is to derive the boundary of the robust rate region for a given channel uncertainty set. We employ the definition of the Pareto boundary for the rate region, which is given next. 
\begin{definition}[\!\!\cite{lameiro2017rate, jorswieck2008complete}]
The rate pair ($R_1,R_2$) is called Pareto-optimal if ($R_1^{\prime},R_2$) and ($R_1,R_2^{\prime}$), with $R_1^{\prime}>R_1$ and $R_2^{\prime}>R_2$ , are not achievable.
\end{definition}

Employing the concept of {\em worst-case robustness} \cite{wang2013robust, feng2015robust,pascual2006robust}, we define the boundary of the robust rate region as the Pareto-optimal points that are achievable for all possible channels inside the uncertainty region \cite{mochaourab2012robust}. Therefore, the robust rate region is the union of all these achievable rate tuples, i.e., 
 \begin{equation}
\mathcal{R}= \underset{\varrho\in\Omega}{\bigcup}(\underset{h_{ij}\in\mathcal{E}_{ij}}{\min}R_1,\underset{h_{ij}\in\mathcal{E}_{ij}}{\min}R_2),\label{rateregion}
\end{equation}
where $\varrho=\!\{p_i,\kappa_i,\phi_i,i=1,2\}$, and $\Omega\!=\!\{p_i,\kappa_i,\phi_i\!:0\leq p_i\leq P_i,0\leq\kappa_i\leq 1,0\leq\phi_i\leq2\pi,\,i=1,2\}$ 
are the design parameters, and the feasible set of the design parameters, 
respectively, with $P_i$ being the power budget of user $i$. 
Hereafter, we represent the worst-case rate of user $i$ by $ R_{i}^w\triangleq\underset{h_{ij}\in\mathcal{E}_{ij}}{\min}R_i$ for $i\in\{1,2\}$.
\begin{figure*}
\setcounter{equation}{6}
\begin{align}
\label{eq-new-8}\kappa_{\bar{\imath}}&=\left\{ \begin{array}{cccc} 0 &\text{if}& \frac{|\tilde{h}_{\bar{\imath}i}|^2(\sigma^2+P_i|\tilde{h}_{i\bar{\imath}}|^2)}{|\tilde{h}_{\bar{\imath}\bar{\imath}}|^2\sigma^2}\leq1-\frac{\gamma_i(1)}{\gamma_i(2\alpha)},&\\
\kappa^{\star} &\text{if}&\frac{|\tilde{h}_{\bar{\imath}i}|^2(\sigma^2+P_i|\tilde{h}_{\bar{\imath}i}|^2)}{|\tilde{h}_{\bar{\imath}\bar{\imath}}|^2\sigma^2}>1-\frac{\gamma_i(1)}{\gamma_i(2\alpha)}&\text{and}\,\,\,\,\,\, \mathcal{P}_{\bar{\imath}}(1)\leq P_{\bar{\imath}},\\
1&&\text{otherwise},&\\
\end{array}\right.
\\
\label{eq-new-9}\mathcal{P}_{\bar{\imath}}(\kappa)&=\left[\sqrt{(1-\frac{\gamma_i(1)}{\gamma_i(2\alpha)})^2+(1-\kappa^2)\left(\frac{\gamma_i(2)}{\gamma_i(2\alpha)}-1\right)}-(1-\frac{\gamma_i(1)}{\gamma_i(2\alpha)})\right]\dfrac{\sigma^2}{|\tilde{h}_{\bar{\imath}i}|^2(1-\kappa^2)},
\end{align}
\hrulefill
\end{figure*}

In order to derive the boundary of the robust rate region, we do not need to optimize over all the six design variables, i.e., the powers, circularity coefficients, and the phases of the complementary variances. 
 According to \eqref{R1}, the achievable rates are functions of the phase difference between the phases of the complementary variances, i.e., $\Delta_{\phi}=\phi_1-\phi_2$. 
The reason is that, as can be observed through (\ref{R1}), $R_i$ depends on the phase parameters only through the term 
\setcounter{equation}{3}
\begin{align}\label{sikim}
\nonumber t&\triangleq\left| p_i\kappa_i|h_{ii}|^2e^{j(2\measuredangle h_{ii}+\phi_i)}+p_{\bar{\imath}}\kappa_{\bar{\imath}}|h_{\bar{\imath}i}|^2e^{j(2\measuredangle h_{\bar{\imath}i}+\phi_{\bar{\imath}})}\right|^2,\\
&=\left| p_i\kappa_i|h_{ii}|^2+p_{\bar{\imath}}\kappa_{\bar{\imath}}|h_{\bar{\imath}i}|^2e^{j(2(\measuredangle h_{\bar{\imath}i}-\measuredangle h_{ii})+\phi_{\bar{\imath}}-\phi_i)}\right|^2
\end{align}
where $i,\bar{\imath}\in\{1,2\}$ and $\bar{\imath}\neq i$. 
Furthermore, it is shown in \cite[Theorem 2]{park2013sinr} that, in every point on the Pareto-optimal boundary of the rate region  for this problem, 
at least one user transmits with maximum power.
As a result, we can further simplify the problem by taking the power of one user equal to its maximum power, and solve two optimization problems with only {\em four} optimization parameters, namely the two circularity coefficients, $\Delta_{\phi}$, and the power of one of the users. 


Let us denote the user that transmits with maximum power by $i$. The boundary of the robust rate region, when user $i$ transmits with maximum power, can be derived by solving
\begin{equation}
 \underset{\Delta_{\phi},\,\,0\leq p_{\bar{\imath}}\leq P_{\bar{\imath}},\,\,0\leq \kappa_1,\kappa_2\leq 1}{\text{maximize}}\,\,\,\,\,\,\,\,  
  R_{\bar{\imath}}^w \,\,\,\,\,\,\,\,\,\,\,\,\,\,\,\, 
  \text{s.t.}  \,\,\,\,\,\,\,\,   
  R_{i}^w\geq\alpha R_{i,\max}^w,
\label{rateregion1}
\end{equation}
\!for $i,\bar{\imath}\in\{1,2\}$ and $i\neq \bar{\imath}$, and a fixed $\alpha\in [0,1]$, where $R_{i,\max}^w$ is the maximum worst-case rate of user $i$, which is obtained with $p_{\bar{\imath}}=0$ and PGS \cite{cover2012elements}. 
That is, we maximize the worst-case rate of user $\bar{\imath}$ for every feasible worst-case rate of user $i$. 
It is worth mentioning that deriving the optimal solution of \eqref{rateregion1} in polynomial time is not possible in general
since \eqref{rateregion1} is not convex  and has an infinite number of non-convex constraints. These infinite number of constraints could be reduced to a single one if an expression for the worst-case channels within the uncertainty sets were identified. However, this is also not possible for arbitrary uncertainty regions. 
To the best of our knowledge, even with perfect CSIT, there are only numerical approaches that provide a suboptimal solution, 
e.g., \cite{zeng2013transmit}, for the rate region of the two-user IC.

\section{Robust design for the two-user IC}\label{IX-ch}

In this section, we propose a robust design for the two-user IC by simplifying the original problem \eqref{rateregion1}. In particular, by allowing only one user to employ IGS, i.e., by setting either $\kappa_1$ or $\kappa_2$ to zero, we can easily find the worst-case channels for arbitrary uncertainty sets. In turn, an optimization problem that approximates \eqref{rateregion1} is obtained, which, although still non-convex, its global optimal solution admits a closed form.

If at most one of the users employs IGS, the achievable rates are independent of the phases of the channel coefficients since $R_i$ depends on the phase parameters only through the term $t$ in \eqref{sikim}. 
Thus, if $\kappa_i=0$ or $\kappa_{\bar{\imath}}=0$, $t$ is independent of the phases of the channel coefficients and as a result,
these phases are not required at the transmitter side in order to optimize the rate. 
Since our proposed robust algorithm is independent of the phases of the channels, 
it requires only the worst-case channel gains and consequently can be applied to any uncertainty model. 
In the following, 
we first derive the worst-case channel gains and then derive the optimal parameters in closed-form. 

It is easy to see that the rate of each user is a strictly
 increasing function of the gain of the corresponding direct link, i.e., $\frac{\partial R_1}{\partial |h_{11}|^2}> 0$ and $\frac{\partial R_2}{\partial |h_{22}|^2}> 0$.
Moreover, the rates of users are  
decreasing functions of the interference channel gain, i.e., $\frac{\partial R_1}{\partial |h_{21}|^2}\leq 0$ and $\frac{\partial R_2}{\partial |h_{12}|^2}\leq 0$.
Thus, the worst-case channel gains are 
$|\tilde{h}_{11}|^{2}=\underset{x\in\mathcal{E}_{11}}{\min}\,\,|x|^2$, $|\tilde{h}_{22}|^{2}=\underset{x\in\mathcal{E}_{22}}{\min}\,\,|x|^2$, $|\tilde{h}_{21}|^{2}=\underset{x\in\mathcal{E}_{21}}{\max}\,\,|x|^2$, and $|\tilde{h}_{12}|^{2}=\underset{x\in\mathcal{E}_{12}}{\max}\,\,|x|^2$.

To derive the boundary of the robust rate region, we have to optimize only two parameters, i.e., the power of the user that may not transmit with maximum power, and the circularity coefficient of the user that may employ IGS.  
Thus, every point of the robust rate region that is achievable by this scheme can be derived by one of the following 
strategies: 
\begin{enumerate}
\item Strategy 1: The PGS user transmits with maximum power,
\item Strategy 2: The IGS user transmits with maximum power.
\end{enumerate} 
Let us denote the rate region achieved by the $k$th strategy as $\mathcal{R}_k$. The robust rate region achieved by the proposed scheme 
is the union of the achievable rate regions of the above strategies, i.e., $\mathcal{R}= \bigcup_{k=1}^2\mathcal{R}_k$. 
It is worth mentioning that $\mathcal{R}_k$ is also a union of two different strategies since there is no difference between users and either of them can be  the IGS user. In the following subsections, we derive the Pareto-optimal achievable rate region of each strategy.

\subsection{Achievable rate region for strategy 1}
In strategy 1 the PGS user transmits with maximum power. Without loss of generality, let us assume that the IGS user is user $\bar{\imath}$. Then, the optimization problem is
\begin{subequations}
\begin{alignat}{4}
 \underset{0\leq p_{\bar{\imath}}\leq P_{\bar{\imath}},\,\,0\leq\kappa_{\bar{\imath}}\leq 1}{\text{maximize}}\,\,\,\,\,\,\,\,  & 
  R_{\bar{\imath}}^w(p_{\bar{\imath}},\kappa_{\bar{\imath}}) \\
  \text{s.t.}  \,\,\,\,\,\,\,\, \,\,\,\,\,\,\,\,\,\,\,\,\,\,\,\,&  
  R_{i}^w(p_{\bar{\imath}},\kappa_{\bar{\imath}})\geq\alpha R_{i,\max}^w,
\end{alignat}
\label{pareto-2}
\end{subequations}
\!for $i,\bar{\imath}\in\{1,2\}$ and $i\neq \bar{\imath}$, and a fixed $\alpha\in [0,1]$, where $R_{i,\max}^w=\log_2(1+\frac{P_i|\tilde{h}_{ii}|^2}{\sigma^2})$. 
Note that the worst-case rates are derived by replacing the worst-case channel gains in \eqref{R1}. 
The  achievable robust rate region can be derived by varying $\alpha\in [0,1]$. 
In \cite{lameiro2015benefits}, a similar scenario in the context of cognitive radio was studied, and \eqref{pareto-2} was solved.
 Thus, we can apply the results in \cite{lameiro2015benefits} to obtain the Pareto-optimal parameters for strategy 1  as presented in the following theorem. 
\begin{theorem}
\label{theorem2}
Let us define $\gamma_i(\alpha)=2^{\alpha R_{i,\max}^w-1}$. 
The Pareto-optimal parameters for transmission strategy 1 are given by \eqref{eq-new-8} and $p_{\bar{\imath}}=\mathcal{P}_{\bar{\imath}}(\kappa_{\bar{\imath}})$, 
where $\kappa^{\star}=\sqrt{1-\frac{\sigma^2}{p_{\bar{\imath}}|\tilde{h}_{\bar{\imath}i}|^2}\left[\left(\frac{\gamma_i(2)}{\gamma_i(2\alpha)}-1\right)\frac{\sigma^2}{p_{\bar{\imath}}|\tilde{h}_{\bar{\imath}i}|^2}-2(1-\frac{\gamma_i(1)}{\gamma_i(2\alpha)})\right]}$, and $\mathcal{P}_{\bar{\imath}}(\kappa)$ is given by  \eqref{eq-new-9}. 
\end{theorem}
\begin{proof}
Please refer to Eq. (11) and Theorem 1 in \cite{lameiro2015benefits} for more details.
\end{proof}

\subsection{Achievable rate region for strategy 2}
In strategy 2 the IGS user transmits with maximum power. Without loss of generality, let us assume that the IGS user is user $\bar{\imath}$. Then, the optimization problem is 
\setcounter{equation}{8}
\begin{subequations}
\begin{alignat}{4}
 \underset{0\leq p_i\leq P_{i},\,\,0\leq \kappa_{\bar{\imath}}\leq 1}{\text{maximize}}\,\,\,\,\,\,\,\,  & 
  R_{i}^w(p_{i},\kappa_{\bar{\imath}}) \\
  \text{s.t.}  \,\,\,\,\,\,\,\, \,\,\,\,\,\,\,\,\,\,\,\,\,\,\,\,&  
 \label{prto-3-9d} R_{\bar{\imath}}^w(p_{i},\kappa_{\bar{\imath}})\geq\alpha R^w_{\bar{\imath},\max},
\end{alignat}
\label{prto-3}
\end{subequations}
\!for $i,\bar{\imath}\in\{1,2\}$ and $i\neq \bar{\imath}$, and a fixed $\alpha\in [0,1]$, where $R_{\bar{\imath},\max}^w=\log_2(1+\frac{P_{\bar{\imath}}|\tilde{h}_{\bar{\imath}\bar{\imath}}|^2}{\sigma^2})$. 
Similar to \eqref{pareto-2}, the robust achievable rate region can be derived by varying $\alpha\in [0,1]$.
The optimization problem \eqref{prto-3} resulting from strategy 2 has not been considered before in the literature. 
We present the Pareto-optimal parameters for this strategy in the subsequent theorem. 

\begin{theorem}\label{th:theorem2}
 The Pareto-optimal  signaling scheme for strategy 2 is IGS if and only if one of the following conditions is met. Additionally, the Pareto-optimal parameters for each condition are provided.  
 \begin{enumerate}
\item $P_i\leq\mathcal{P}_i(1)$, $\Rightarrow p_i=P_i$ and $\kappa_{\bar{\imath}}=1$,
  \item $P_i>\mathcal{P}_i(1)$, $(\zeta_1\beta_2-\zeta_2\beta_1)<0$ and $\tau<0$, $\Rightarrow p_i=\mathcal{P}_i(1)$ and $\kappa_{\bar{\imath}}=1$,
  \item $P_i>\mathcal{P}_i(1)$, $(\zeta_1\beta_2-\zeta_2\beta_1)<0$, $\tau>0$, and $\mathcal{P}(0)>x_1^{\star}$, $\Rightarrow p_i=\max(x_1^{\star},\mathcal{P}_i(1))$ and $\kappa_{\bar{\imath}}=\mathcal{K}(p_i),$
  \item $P_i>\mathcal{P}_i(1)$, $(\zeta_1\beta_2-\zeta_2\beta_1)>0$, $\tau<0$, $\mathcal{P}_i(1)<x_2^{\star}$, and $R^w_i(\mathcal{P}(1))>R^w_i(\mathcal{P}(0))$, $\Rightarrow p_i=\mathcal{P}_i(1)$ and $\kappa_{\bar{\imath}}=1$,
\end{enumerate}
where $i$ and $\bar{\imath}$ are the users that employ PGS and IGS in the transmission strategy 2, respectively, $\gamma_{\bar{\imath}}(\cdot)$ is defined as in Theorem \ref{theorem2}, and
\begin{align}
\label{eqo-304}x^{\star}_1&=\frac{-\zeta_1\tau-\sqrt{(\zeta_1\tau)^2-\beta_1\tau(\zeta_1\beta_2-\zeta_2\beta_1)}}{\zeta_1\beta_2-\zeta_2\beta_1},\\
\label{eqo-304-1}x^{\star}_2&=\frac{-\zeta_1\tau+\sqrt{(\zeta_1\tau)^2-\beta_1\tau(\zeta_1\beta_2-\zeta_2\beta_1)}}{\zeta_1\beta_2-\zeta_2\beta_1},\\
\nonumber \mathcal{P}_i(\kappa)&=\frac{1}{|\tilde{h}_{i\bar{\imath}}|^2}\\
&\hspace{0.5cm}\left[\left(\frac{1+\sqrt{1+\gamma_{\bar{\imath}}(2\alpha)(1-\kappa^2)}}{\gamma_{\bar{\imath}}(2\alpha)}\right)P_{\bar{\imath}}|\tilde{h}_{\bar{\imath}\bar{\imath}}|^2-\sigma^2\right],\label{eqo-3050}\\
\label{eqo-305}\mathcal{K}(p_i)&=\sqrt{1-\frac{(\sigma^2+p_i|\tilde{h}_{i\bar{\imath}}|^2)^2}{P_{\bar{\imath}}^2|\tilde{h}_{\bar{\imath}\bar{\imath}}|^4}\gamma_{\bar{\imath}}(2\alpha)+2\frac{\sigma^2+p_i|\tilde{h}_{i\bar{\imath}}|^2}{P_{\bar{\imath}}|\tilde{h}_{\bar{\imath}\bar{\imath}}|^2}},\\
\label{r-opt-2}
R^w_i(p_i)&= \frac{1}{2}\!\log_2\!\left(1+\!\frac{\zeta_1p_i^2+\beta_1p_i}{\zeta_2p_i^2+\beta_2p_i+\tau}\!\right).
\end{align}
Moreover,  $\beta_1$, $\zeta_1$,  $\beta_2$, $\zeta_2$ and $\tau$ are 
\begin{align}\label{zeta}
\nonumber \beta_1&=2|\tilde{h}_{ii}|^2(\sigma^2+P_2|\tilde{h}_{\bar{\imath}i}|^2),
\\ \nonumber \beta_2&=2(\sigma^2\gamma_{\bar{\imath}}(2\alpha)-P_{\bar{\imath}}|\tilde{h}_{\bar{\imath}\bar{\imath}}|^2)\frac{|\tilde{h}_{\bar{\imath}i}|^4|\tilde{h}_{i\bar{\imath}}|^2}{|\tilde{h}_{\bar{\imath}\bar{\imath}}|^4},
\\
\nonumber \zeta_1&=|\tilde{h}_{ii}|^4,\\
 \nonumber \zeta_2&=\frac{|\tilde{h}_{\bar{\imath}i}|^4|\tilde{h}_{i\bar{\imath}}|^4}{|\tilde{h}_{\bar{\imath}\bar{\imath}}|^4}\gamma_{\bar{\imath}}(2\alpha),\\
\tau&=\!\sigma^4\!\!+\!2\sigma^2P_{\bar{\imath}}|\tilde{h}_{\bar{\imath}i}|^2\!-\!2\sigma^2P_{\bar{\imath}}\frac{|\tilde{h}_{i\bar{\imath}}|^2}{|\tilde{h}_{\bar{\imath}\bar{\imath}}|^2}+\sigma^4\gamma_{\bar{\imath}}(2\alpha)\frac{|\tilde{h}_{\bar{\imath}i}|^4}{|\tilde{h}_{\bar{\imath}\bar{\imath}}|^4}.&&
\end{align}
\end{theorem}
\begin{proof}
Please refer to Appendix \ref{app:theorem2}.
\end{proof}
Note that if none of the conditions in Theorem \ref{th:theorem2} is fulfilled, PGS is Pareto-optimal for both users in strategy 2. 
Moreover, the parameters of the users can be easily derived through \eqref{pareto-2} by taking $\kappa_{\bar{\imath}}=0$.

The implication of Theorem \ref{th:theorem2} can be understood with the following example.
If, e.g., user 2 employs IGS, it causes less interference to user 1. 
Thus, user 1 can decrease its transmission power in order to meet the rate constraint of user 2 in \eqref{prto-3-9d}.
According to Theorem \ref{th:theorem2}, this power reduction of user 1, alongside with employing IGS by user 2, may even result in a rate increase for user 1. 
In other words, IGS allows the users to decrease the transmission power and simultaneously increase the achievable rate,
hence, improving as well the power efficiency.  
This is due to the fact that IGS can mitigate the negative effect of the 
interference, which 
may lead to an overall improvement of the system performance.

\section{Robust design for the Z-IC}\label{secIII}
The results in Section \ref{IX-ch} can be applied to the Z-IC by taking $h_{12}=0$. 
However, for this simplified scenario it is possible to obtain 
a better closed-form robust design for the Z-IC if we allow both users to employ IGS. 
When both users employ IGS, the phases of the channels are relevant  for the optimal strategy. Thus, even though problem  \eqref{rateregion1} is simpler for the Z-IC, it is still in general difficult to obtain the worst-case channels when both magnitude and phase are considered. Therefore, the constraint set of problem \eqref{rateregion1} cannot be reduced to a finite number of constraints and hence it is still difficult to solve in its current form. To overcome this, we approximate the original problem by considering a surrogate uncertainty region in which magnitude and phase are decoupled, so that their worst realizations can easily be found. 
It is also worth mentioning that this approach gives us a lower bound for the worst-case rates of problem \eqref{rateregion1}. 
Note that the surrogate uncertainty region must contain the original region as illustrated in Fig. \ref{Fig--2} in order for the constraints of \eqref{rateregion1} to be fulfilled after solving the approximated problem. Thus, we consider the enlarged uncertainty region 
 $\tilde{\mathcal{E}}_{ij}\!=\!\{x\in\mathbb{C}\!:|x|=|\hat{h}_{ij}|+e_{|h_{ij}|},\measuredangle x=\measuredangle \hat{h}_{ij}+ e_{\measuredangle h_{ij}}, |e_{|h_{ij}|}|\leq\delta_{ij},|e_{\measuredangle h_{ij}}|\leq\theta_{ij} \}$, where 
 $\delta_{ij}$ and $\theta_{ij}$ are the largest uncertainties in magnitude and phase, respectively.

The enlarged uncertainty region permits decoupling the errors in phase and magnitude 
 and hence allows us to find the worst-case phases and the worst-case channel gains independently. 
It is worth mentioning that, since the errors in phase and magnitude are not necessarily independent, a channel realization with both the worst-case channel phase and worst-case channel gains might not be in the original set. Hence, this approach provides a lower bound for the worst-case rates of \eqref{rateregion1}.
Note that the enlarged uncertainty region includes the original region as a subset, and the bounds of the errors in phase and magnitude are the same for both uncertainty regions. 
Thus, this approach can be applied to any arbitrary uncertainty model.

\begin{figure}[t]
\centering
\includegraphics[width=0.36\textwidth]{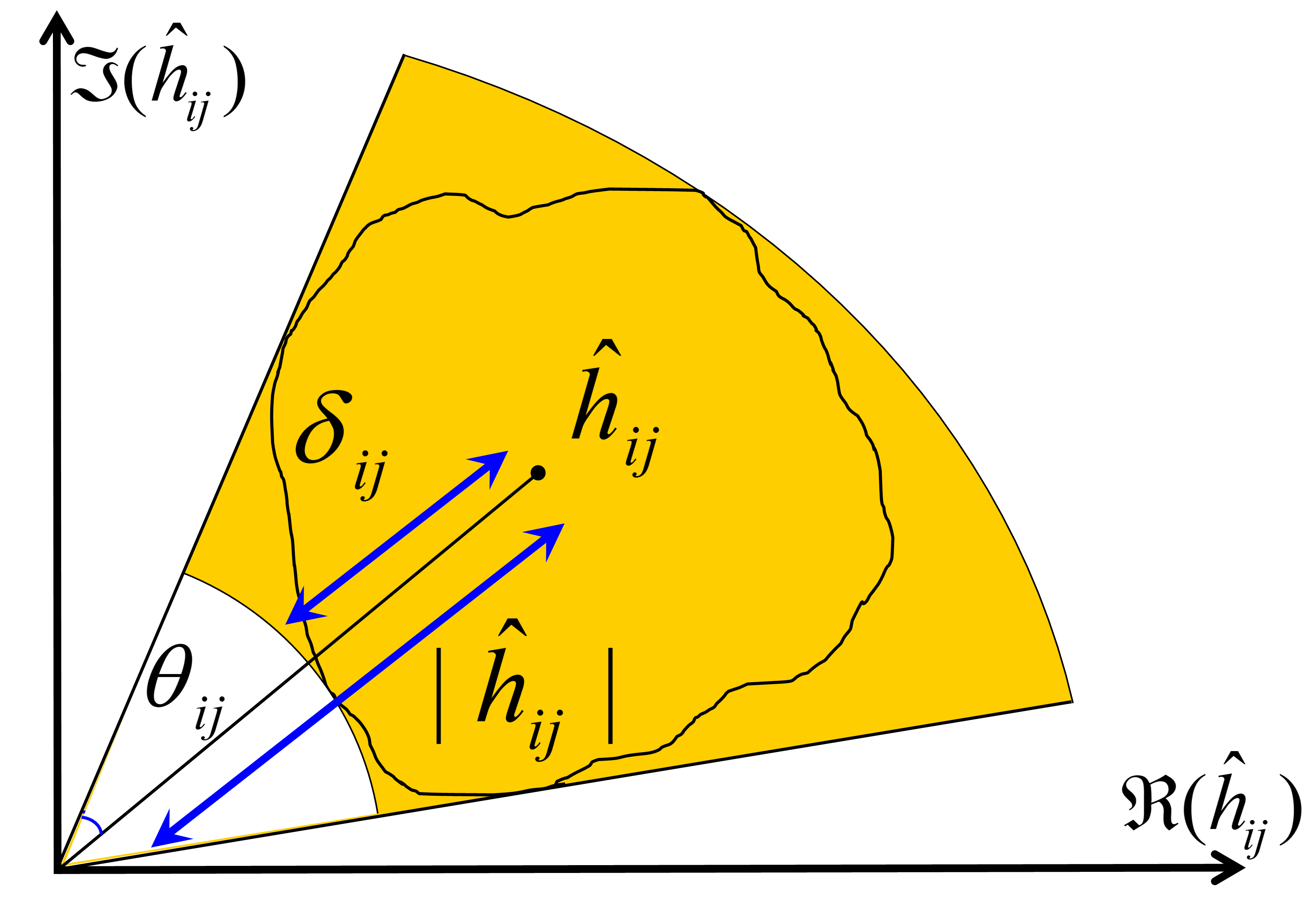}
\caption{Enlarged channel uncertainty region for the uncertainty region in Fig \ref{Fig2-1}.}
\label{Fig--2}
\end{figure}

In the Z-IC, since user 1 does not interfere with user 2, the optimal design parameters for user 1 are those maximizing its rate.
 Thus, every point in the Pareto boundary of the 
 robust rate region can be achieved when user 1 transmits with maximum power, i.e., $p_1=P_1$.  
Hence, in the rest of this section, we consider $p_1=P_1$. 
The rate of user 2 in the Z-IC can be derived by taking $h_{12}=0$ in \eqref{R1} as
\begin{equation}\label{e15}
R_2=\frac{1}{2}\log_2\left(\frac{(p_2|h_{22}|^2+\sigma^2)^2-(\kappa_2p_2|h_{22}|^2)^2}{\sigma^4}\right).
\end{equation}
Moreover, the robust rate region of the Z-IC for the enlarged uncertainty region can be derived by replacing $\mathcal{E}_{ij}$ with $\tilde{\mathcal{E}}_{ij}$ 
in \eqref{rateregion} as
\begin{equation}
\tilde{\mathcal{R}}= \underset{\varrho\in\Omega}{\bigcup}(\underset{h_{ij}\in\tilde{\mathcal{E}}_{ij}}{\min}R_1,\underset{h_{ij}\in\tilde{\mathcal{E}}_{ij}}{\min}R_2),\label{rr-enl}
\end{equation}
where $\tilde{\mathcal{R}}\subset\mathcal{R}$.
In the following, we first derive the worst-case channel coefficients in the enlarged uncertainty region, which are the solution of the minimization part in \eqref{rr-enl}. 
Then, we will derive the transmission parameters that attain the boundary of the robust rate region for the enlarged uncertainty region, which is the solution of the optimization problem 
\begin{equation}
 \underset{\Delta_{\phi},0\leq p_2\leq P_2,0\leq\kappa_1,\kappa_2\leq 1}{\text{maximize}}\,\,\,\,\,\,\,\,  
  R_{2}^w, \,\,\,\,\,\,\,\,\,\,\,\,\,\,\,\,
   \text{s.t.}  \,\,\,\,\,\,\,\,   
  R_{1}^w\geq\alpha R_{1,\max}^w,
\label{rr-en}
\end{equation}
\!for a fixed $\alpha\in [0,1]$. The robust rate region can be derived by varying  $\alpha\in [0,1]$.

The worst-case gains are the same as those derived in Section \ref{IX-ch}, i.e., the minimum channel gain in the uncertainty region $\tilde{\mathcal{E}}_{ij}$ for the direct links, and the maximum channel gain in $\tilde{\mathcal{E}}_{ij}$ for the interference links. 
We rewrite them here as 
$|\tilde{h}_{11}|^{2}=\underset{x\in\mathcal{E}_{11}}{\min}\,\,|x|^2$, $|\tilde{h}_{22}|^{2}=\underset{x\in\mathcal{E}_{22}}{\min}\,\,|x|^2$, $|\tilde{h}_{21}|^{2}=\underset{x\in\mathcal{E}_{21}}{\max}\,\,|x|^2$, and $|\tilde{h}_{12}|^{2}=\underset{x\in\mathcal{E}_{12}}{\max}\,\,|x|^2$.

In the following, we derive the worst-case phase error along with the Pareto-optimal phases $\phi_1$ and $\phi_2$ for the given worst-case channel gains. Through (\ref{R1}), it can be observed that $R_1$ depends on these quantities only through the term 
\begin{equation}
A\triangleq\left| p_1\kappa_1|\tilde{h}_{11}|^2e^{j(2\measuredangle h_{11}-2\measuredangle h_{21}+\phi_1-\phi_2)}+p_2\kappa_2|\tilde{h}_{21}|^2\right|^2,
\end{equation}
while the rate of user 2 is independent of them (see \eqref{e15}).   
Replacing the true values by the estimated values, we have $2\measuredangle h_{11}\!-\!2\measuredangle h_{21}\!=\!2\measuredangle \hat{h}_{11}\!-\!2\measuredangle \hat{h}_{21}\!\pm 2
e_{\measuredangle h_{11}}\pm 2
e_{\measuredangle h_{21}}$. Let us define
$\Delta_{ch}\triangleq 2
e_{\measuredangle h_{11}}+ 2
e_{\measuredangle h_{21}}\in [-\theta,\theta]$ as the aggregate uncertainty in phase, where $\theta=2\theta_{11}+2\theta_{21}$ is the magnitude of the maximum aggregate phase error. 
Note that  $\theta=\pi$ means that there is no reliable information about the phase of the channels.
Finally, the term $A$ can be written as 
\begin{multline}
A\!=\! p_2^2\kappa_2^2|\tilde{h}_{21}|^4\!+2p_1\kappa_1p_2\kappa_2|\tilde{h}_{21}|^2|\tilde{h}_{11}|^2\!\cos(\Delta_{ch}+\Delta_{\phi^{'}}\!)
\\+p_1^2\kappa_1^2|\tilde{h}_{11}|^4,\label{eq:A}
\end{multline}
where $\Delta_{\phi^{'}}=(\phi_1+2\measuredangle \hat{h}_{11})-(\phi_2+2\measuredangle \hat{h}_{21})$.
In the following lemma, we state the Pareto-optimal phase parameters.

\begin{lemma}
In the Z-IC, each point of the boundary of the rate region defined in \eqref{rr-enl} can be achieved by $\phi_1=0$ and $\phi_2=2\measuredangle \hat{h}_{11}-2\measuredangle \hat{h}_{12}+\pi$. Furthermore, the corresponding worst-case channel phases yield $\Delta_{ch}^{\star}=\theta$.
\end{lemma}
\begin{proof}
Since $\Delta_{ch}$ and $\Delta_{\phi^{'}}$ only affect $R_1$, their values describing the boundary of the rate region defined in \eqref{rr-enl} can be obtained as the solution of the maximin problem
\begin{align}
\nonumber (\Delta_{\phi^{'}}^{\star},\Delta_{ch}^{\star})&=\arg\underset{\Delta_{\phi^{'}}}{\max} \,
\underset{\Delta_{ch}}{\min}(R_1)\overset{\mathrm{(*)}}{=} \arg\underset{\Delta_{\phi^{'}}}{\min} \,
\underset{\Delta_{ch}}{\max}(A)
\\&=
\arg\underset{\Delta_{\phi^{'}}}{\min} \,
\underset{\Delta_{ch}}{\max}(\cos(\Delta_{\phi^{'}}+\Delta_{ch})).
\label{phi}
\end{align}
\begin{figure}[t]
\centering
\includegraphics{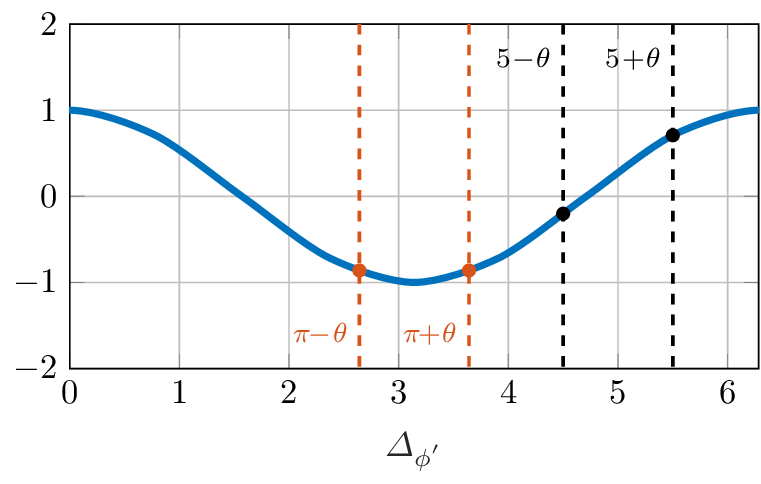}
\caption{Function $\cos(\Delta_{\phi^{'}}+\Delta_{ch})$ for two examples with $\theta=0.5$, $\Delta_{\phi^{'}}=\pi$ (red) and $\Delta_{\phi^{'}}= 5$ (black). }
\label{cos}
\end{figure}
The equality $(*)$ in \eqref{phi} holds since $R_1$ is a function of the phases only through $A$. 
Moreover, $R_1$ decreases with $A$.
In order to solve this problem, we consider $\Delta_{\phi^{'}}$ in a period, i.e., $0\leq\Delta_{\phi^{'}}< 2\pi$, as depicted in Fig \ref{cos}.
Thus, the solution of the maximization problem in \eqref{phi} is 
\begin{equation}
\Delta_{ch}^{\star}=\left\{ \begin{array}{rcl}
\max(-\Delta_{\phi^{'}},-\theta) & \mbox{for}
& 0\leq\Delta_{\phi^{'}}< \pi, \\ 
\min(2\pi-\Delta_{\phi^{'}},\theta) & \mbox{for}
& \pi\leq\Delta_{\phi^{'}}< 2\pi.
\end{array}\right.
\end{equation}
 This is because $\cos(\phi)$ is decreasing in the interval $[0,\pi]$, increasing in the interval $[\pi,2\pi ]$ and maximized at $\phi=0$, or $\phi=2\pi$. 
Moreover, since $\cos(\phi)$ is decreasing (increasing) in the interval $[0,\pi]$ ($[\pi,2\pi]$),  it can be easily seen through Fig. \ref{cos} that $\Delta_{\phi^{'}}^{\star}$ should be such that $\Delta_{\phi^{'}}+\Delta_{ch}$ is as close as possible to $\pi$, which results in $\Delta_{\phi^{'}}=\pi$.
In other words, the solution of the minimax problem 
 is $\Delta_{\phi^{'}}^{\star}=\pi$ and $\Delta_{ch}^{\star}=\theta$. 
Through \eqref{e15} and \eqref{R1}, we can see that the rate of user 2 is independent of the phase parameters, and the rate of user 1 is only a function of the phase difference. Thus, we can, without loss of generality,  choose $\phi_1=0$ and $\phi_2=2\measuredangle \hat{h}_{11}-2\measuredangle \hat{h}_{21}+\pi$.
\end{proof}
\begin{figure*}
\setcounter{equation}{25}
\begin{align}
\label{new-eq-26}
R_1^w(p_2,\kappa_2)\!=&\!\left\{\!\!\begin{array}{lcr} \!\frac{1}{2}\log_2\!\left(\!1\!+\frac{2P_1|\tilde{h}_{11}|^2\left(\sigma^2+p_2|\tilde{h}_{21}|^2(1+\kappa_2\cos\theta)\right)}{p^2_2|\tilde{h}_{21}|^4(1-\kappa^2_2)+2\sigma^2p_2|\tilde{h}_{21}|^2+\sigma^4}\!\!\right)
&\text{if}&\kappa_1^{\star}=1,\\
\!\frac{1}{2}\log_2\!\left(\!\frac{\left(P_1|\tilde{h}_{11}|^2+p_2|\tilde{h}_{21}|^2+\sigma^2\right)^2-p_2^2|\tilde{h}_{21}|^4\kappa_2^2\sin^2\theta}{p^2_2|\tilde{h}_{21}|^4(1-\kappa^2_2)+2\sigma^2p_2|\tilde{h}_{21}|^2+\sigma^4}\!\right)&\text{if}&\kappa_1^{\star}<1, \end{array}\right.\\
\label{r11}
R_2^w(p_2,\kappa_2)\!=&\frac{1}{2}\log_2\!\left(p_2^2|\tilde{h}_{22}|^4\sigma^{-4}(1\!-\!\kappa_2^2)+2p_2|\tilde{h}_{22}|^2\sigma^{-2}\!+\!1\right).
\end{align}
\setcounter{equation}{28}
\begin{equation}\label{q2}
q(\kappa_2,\theta)\!=\!\!\left\{\!\!\begin{array}{c}\!\!\! \frac{P_1|\tilde{h}_{11}|^2-\sigma^2(\gamma_1(2\alpha)+1)+\sqrt{(P_1|\tilde{h}_{11}|^2-\sigma^2\gamma_1(2\alpha))^2-\eta\left(\sigma^4(\gamma_1(2\alpha)+1)-2P_1|\tilde{h}_{11}|^2\sigma^2-P_1^2|\tilde{h}_{11}|^4\right)}}{|\tilde{h}_{21}|^2\eta}\\\text{if}\hspace{1cm}\kappa_1\!<\!1,\\
\!\!\frac{P_1|\tilde{h}_{11}|^2(1\!+\kappa_2\cos\theta)\!-\sigma^2\gamma_1(2\alpha)+\sqrt{P_1^2|\tilde{h}_{11}|^4(1+\kappa_2\cos\theta)^2-2P_1|\tilde{h}_{11}|^2\gamma_1(2\alpha)\sigma^2\kappa_2(\cos\theta+\kappa_2)+\gamma_1^2(2\alpha)\kappa_2^2\sigma^4}}{|\tilde{h}_{21}|^2\gamma(2\alpha)(1-\kappa_2^2)}\\\text{if}\hspace{1cm}\kappa_1\!=\!1, \end{array}\right.
\end{equation}
\hrulefill
\end{figure*}

We now derive the optimal transmit powers and circularity coefficients for the worst-case channels, so that the boundary of the rate region \eqref{rr-enl} is attained. 
Since user 1 does not interfere with user 2, the optimal design parameters for user 1 maximize its rate.
 Thus, its circularity coefficient can be obtained as 
\setcounter{equation}{22}
\begin{equation}\label{opt}
\kappa_1^{\star}=\arg\,\underset{\kappa_1}{\max}\left(R_1^w\right)=\arg\,\underset{\kappa_1}{\min}\left(A^{\star}\right),
\end{equation}
where $A^{\star}$ is obtained by taking the worst-case channels in \eqref{eq:A}. Since $A^{*}$ is convex in $\kappa_1$, the solution of \eqref{opt} can be derived as
\begin{equation}
\frac{\partial A^{\star}}{\partial \kappa_1}=-2P_1p_2\kappa_2|\tilde{h}_{21}|^2|\tilde{h}_{11}|^2\cos\theta+2P_1^2\kappa_1|\tilde{h}_{11}|^4=0.
\end{equation}
Taking the feasible set of $\kappa_1$ into account, we obtain
\begin{equation}\label{k1}
\kappa_1^{\star}=\min\left(1,\left[\frac{p_2|\tilde{h}_{21}|^2}{P_1|\tilde{h}_{11}|^2}\kappa_2\cos\theta\right]^+\right).
\end{equation}
From this equation we can readily observe that, if the total uncertainty in phase is equal to or greater than $\frac{\pi}{2}$ (i.e., $\theta\geq\frac{\pi}{2}$), user 1 should transmit proper Gaussian signals. In such a case, the resulting problem becomes equivalent to the proposed robust algorithm in Section \ref{IX-ch}-A, in which only one of the users employs IGS. 
Thus, in the following, we assume that $\theta<\pi/2$ 
and
derive a condition for the optimality of IGS for user 2 in terms of $\theta$, as well as its transmission parameters.

Plugging into \eqref{R1} the optimal transmission parameters of user 1 yields the worst-case rates in \eqref{new-eq-26} and \eqref{r11}.
The Pareto-optimal boundary for the enlarged uncertainty region can then be obtained by rewriting \eqref{rr-en} as 
\setcounter{equation}{27}
\begin{subequations}
\begin{alignat}{4}
 \underset{ 0\leq p_2\leq P_2, 0\leq\kappa_2\leq 1}{\text{maximize}}\,\,\,\,\,\,\,\,  & 
  R^w_{2}(p_2,\kappa_2), \\
  \text{s.t.}  \,\,\,\,\,\,\,\, \,\,\,\,\,\,\,\,\,\,\,\,\,\,\,\,& 
 \label{pareto-28-d} R^w_{1}(p_2,\kappa_2)\geq\alpha R^w_{1,\max},
\end{alignat}
\label{pareto}
\end{subequations}
\!\!where $\alpha\in [0,1]$ is fixed and $R_{1,\max}^w$ is the maximum worst-case achievable rate for user 1, which is obtained with $p_2=0$ and PGS \cite{cover2012elements}. By varying $\alpha$ between 0 and 1, the solution of (\ref{pareto}) provides every point of the robust rate region boundary in \eqref{rr-enl} \cite{lameiro2017rate}.  
Unfortunately, the optimization problem (\ref{pareto}) is not convex due to constraint \eqref{pareto-28-d}. 
Furthermore, the results in \cite{lameiro2017rate} for perfect CSIT cannot be applied due to the phase error $\theta$.
In the following lemma, we rewrite constraint \eqref{pareto-28-d} in a more convenient form to simplify the optimization problem.

\begin{lemma}\label{th:lemma1}
The constraint \eqref{pareto-28-d} is simplified to $p_2\leq q(\kappa_2,\theta)$, where $q(\kappa_2,\theta)$ is given by \eqref{q2}, 
where $\eta=\left((\gamma(2\alpha)+1)(1-\kappa^2_2)-1+\kappa^2_2\sin^2\theta\right)$ and $\gamma_1(x)\triangleq2^{x R^w_{1,\max}}-1$. 
\end{lemma}
\begin{proof}
Constraint \eqref{pareto-28-d} can be simplified to a quadratic function of $p_2$ by plugging \eqref{r11} into \eqref{pareto-28-d} as 
\setcounter{equation}{29}
\begin{equation}\label{q22}
\begin{array}{ccc} 
p_2^2|\tilde{h}_{21}|^4\eta\!-2P_1|\tilde{h}_{11}|^2\sigma^2\!-\!P_1^2|\tilde{h}_{11}|^4+\!\sigma^4(\gamma_1(2\alpha)\!+\!1)\\
-\!2p_2|\tilde{h}_{21}|^2\left[P_1|\tilde{h}_{11}|^2\!-\!\sigma^2\gamma_1(2\alpha)\right]\!\leq 0\hspace{.7cm}\text{if}\hspace{.7cm}\kappa_1\!<\!1,\\
p^2_2|\tilde{h}_{21}|^4(1\!-\!\kappa^2_2)\gamma_1(2\alpha)+\!\gamma_1(2\alpha)\sigma^4\!-\!2P_1\sigma^2|\tilde{h}_{11}|^2\\
+2p_2|\tilde{h}_{21}|^2\left[\sigma^2\gamma_1(2\alpha)-P_1|\tilde{h}_{11}|^2(1+\kappa_2\cos\theta)\right]\!

\!\leq \!0\\\hspace{.7cm}\text{if}\hspace{.7cm}\kappa_1\!=\!1. \end{array}
\end{equation}
We consider (\ref{q22}) as two quadratic polynomials in $p_2$ and take the positive root in (\ref{q22}), which results in (\ref{q2}).
\end{proof}

According to Lemma \ref{th:lemma1}, we can combine 
the power constraint and \eqref{pareto-28-d} into $0\leq p_2(\theta)\leq \min\{P_2,q(\kappa_2,\theta)\}$. Similar to \cite{lameiro2017rate}, we must set $p_2(\theta)=\min\{P_2,q(\kappa_2,\theta)\}$ in order to achieve the global optimum of (\ref{pareto}). 
As a result, the rate of user 2 is only a function of $\kappa_2$ as 
\begin{multline}
R^w_2(\kappa_2,\theta)=\frac{1}{2}\log_2\left(\frac{p_2^2(\kappa_2,\theta)|\tilde{h}_{22}|^4(1-\kappa_2^2)}{\sigma^{4}}\right.\\
\left.+\frac{2p_2(\kappa_2,\theta)|\tilde{h}_{22}|^2}{\sigma^{2}}+1\right).
\end{multline}
Finally, the optimization problem in (\ref{pareto}) is simplified to
\begin{equation}\label{opt2}
 \underset{0\leq\kappa_2\leq 1}{\text{maximize}}\,\,\,\,\,\,\,\,   
  R^w_{2}(\kappa_2,\theta).
\end{equation}
Notice that the difference between this problem and the one solved in \cite{lameiro2017rate} is in the error term $\theta$, which makes its solution not straightforward as the results in \cite{lameiro2017rate} are not applicable. 
Taking $\theta=0$ makes \eqref{opt2} equivalent to the problem considered in \cite{lameiro2017rate}.

\begin{theorem}\label{th:theorem1}
The optimal transmission strategy for user 2 is IGS if $P_2>q(0,\theta)$ and
\begin{equation}\label{k2}
\frac{|\tilde{h}_{21}|^2}{|\tilde{h}_{22}|^2}> \frac{\sigma^2\gamma_1(2\alpha)-P_1|\tilde{h}_{11}|^2-(\frac{P_1|\tilde{h}_{11}|^2}{\gamma_1(\alpha)}-\sigma^2)\cos^2\theta}{\sigma^2\left(\gamma_1(2\alpha)+\cos^2\theta\right)}.
\end{equation}
In this case, the optimal circularity coefficient of user 2 is
\begin{equation}\label{optk22}
\kappa=\left\{\begin{array}{ll}1&\text{if}\,\,q(1,\theta)\leq P_2,\\\kappa_{\max}&\text{otherwise},\end{array}\right.
\end{equation}
where $\kappa_{\max}$ is the minimum value of $\kappa_2$ that results in $P_2\leq q(\kappa_{\max},\theta)$ and $\gamma_1(x)\triangleq2^{x R^w_{1,\max}}-1$.
\end{theorem}
\begin{proof}
Please refer to Appendix \ref{app:theorem1}. 
\end{proof}
\begin{remark}
According to Theorem \ref{th:theorem1}, even if there is no reliable phase information, i.e., when $\theta$ is large, the optimal transmission strategy for user 2 may be IGS if the interference level is sufficiently high, i.e., if \eqref{k2} holds. 
\end{remark}
Notice that the condition presented in Theorem \ref{th:theorem1} is sufficient, but may not be necessary. As shown in Appendix \ref{app:theorem1}, condition \eqref{k2} is obtained by showing that $R_2^w(\kappa_2,\theta)$ is increasing at $\kappa_2=0$ if and only if this condition holds. 
However, IGS might provide a minor gain when  \eqref{k2} does not hold. 
As shown in \cite{lameiro2017rate}, there is only a small performance advantage of IGS in perfect CSI when $R_2^w(\kappa_2,\theta)$ is not increasing around zero, which vanishes as $\kappa_1$ approaches 0. 
Therefore, the solution presented in Theorem \ref{th:theorem1} provides an almost-optimal characterization of the Pareto boundary for the enlarged uncertainty region.

\section{Numerical Results}\label{secIV}

In this section we illustrate our findings with some numerical examples. 
Throughout this section, for the sake of illustration, 
we consider a proper complex Gaussian distribution for the aggregate CSI error similar to the models in \cite{pascual2006robust,wei2017optimal,chang2014radio}. 
Nevertheless, our proposed design can be applied to any uncertainty set, as mentioned before.
We consider the uncertainty region as the region where the true channel lies with probability $\psi$.
Thus, the uncertainty region for the SISO channel is a circle centered at $\hat{h}_{ij}$, as illustrated in Fig. \ref{Fig2}. The radius of this circle is 
\begin{equation}\label{radius}
\delta_{ij}=
\sqrt{-\sigma^2_{ij}\ln(1-\psi)},
\end{equation}where $\sigma^2_{ij}$ is the variance of the channel estimation error. 
Hence, the uncertainty set for channel $h_{ij}$ is $\mathcal{E}_{ij}\!=\!\{x\in\mathbb{C}\!:\! x\!=\!\hat{h}_{ij}+ e_{ij},  | e_{ij}|^2 \leq \delta_{ij}^2\}$. 
Moreover, the enlarged uncertainty region is 
 $\tilde{\mathcal{E}}_{ij}\!=\!\{x\in\mathbb{C}\!:|x|=|\hat{h}_{ij}|+e_{|h_{ij}|},\measuredangle x=\measuredangle \hat{h}_{ij}+ e_{\measuredangle h_{ij}}, |e_{|h_{ij}|}|\leq\delta_{ij},|e_{\measuredangle h_{ij}}|\leq\theta_{ij} \}$, where $\theta_{ij}=\arcsin\frac{\delta_{ij}}{|\hat{h}_{ij}|}$ as depicted in Fig. \ref{Fig2}. 
\begin{figure}[t]
\centering
\includegraphics[width=0.34\textwidth]{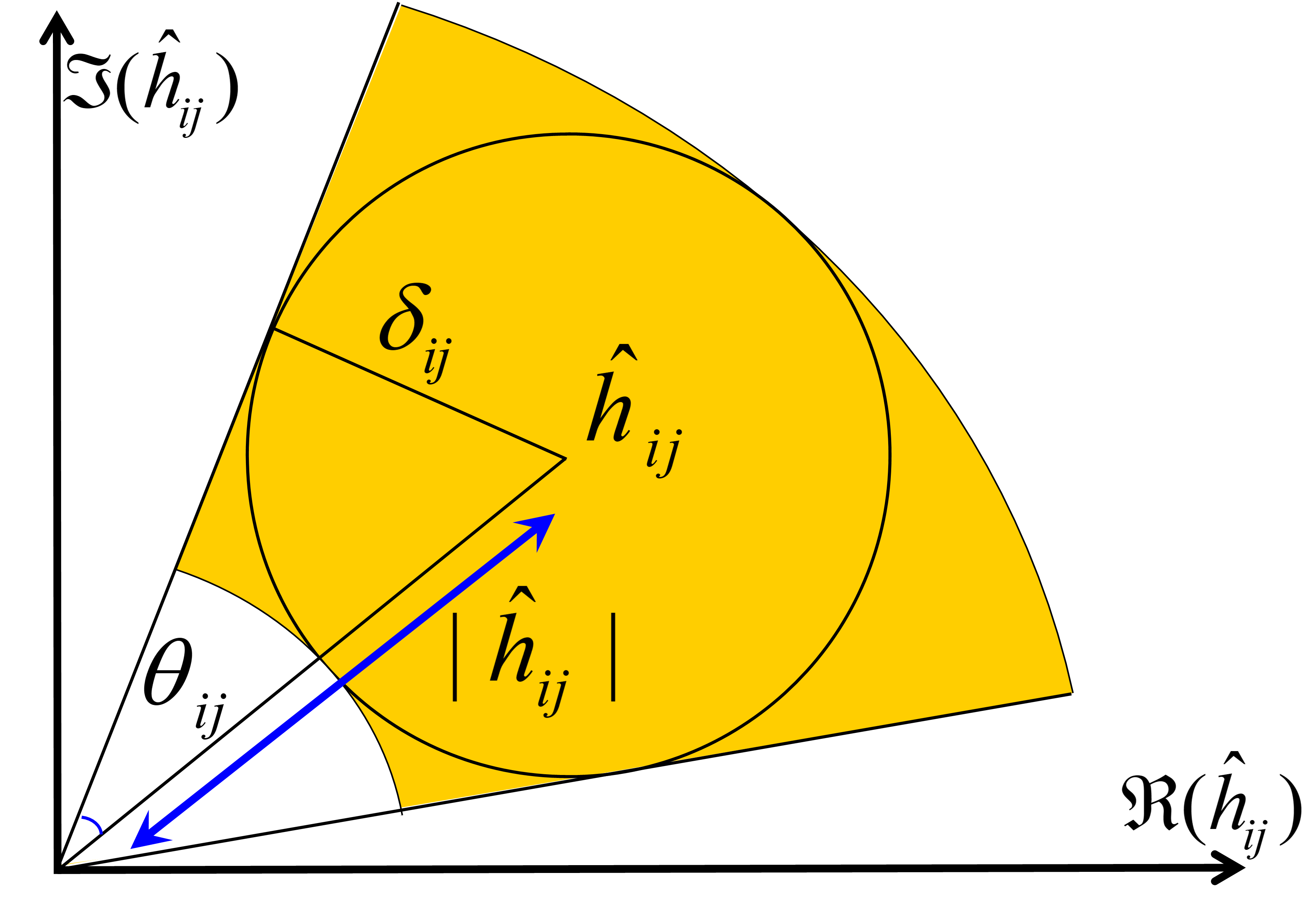}
\caption{Channel uncertainty and enlarged channel uncertainty regions when the CSI errors are modeled as proper Gaussian.}
\label{Fig2}
\end{figure}

In this section, we consider $\sigma^2=1$ and $P=P_1=P_2$, unless it is explicitly mentioned. 
We also assume 
the same variance of the estimation error for all links, i.e., $\sigma^2_{e_{ij}}=\sigma^2_{e}$ for $i,j\in\{1,2\}$ unless it is explicitly mentioned. 
That yields the same size of the uncertainty region for all channels, i.e., $\delta_{ij}=\delta$ for $i,j\in\{1,2\}$, where $\delta$ is given by (\ref{radius}). 
Furthermore, we compare our proposed robust designs with the existing non-robust designs \cite{lameiro2017rate} and  \cite{zeng2013transmit} for the Z-IC and two-user IC, respectively.
We choose the joint covariance and complementary covariance algorithm in \cite{zeng2013transmit} as a non-robust design for two-user IC. 
Moreover, the robust PGS is the PGS design for the worst-case channels gains.

In the following subsections we provide two different types of numerical examples. 
In the first type, we provide results averaged over a large number of channel realizations for a specific point of the rate region. 
For example, we derive the average sum-rate for a specific point of the rate region 
averaged  over 200 channel realizations in Figs. \ref{Fig80}, \ref{Fig9} and \ref{Fig10}. 
For this type of numerical examples, we employ the Monte Carlo method and generate the channels randomly. 
In each channel realization, each channel estimate is drawn from a proper complex Gaussian distribution with zero mean and unit variance.
In the second type, we derive the whole rate region for a specific channel realization.
It is worth mentioning that both types of numerical examples are common for rate region analysis (please refer to \cite{zeng2013transmit, lameiro2017rate,lagen2016superiority}).

\subsection{Results for the two-user IC}
In this subsection, we present the results for the two-user IC. 
The proposed and existing techniques are denoted as follows. 
\begin{itemize}
\item {\bf R-IGS}: Our proposed robust IGS in Section \ref{IX-ch}.
\item {\bf R-PGS}: Robust PGS.
\item {\bf N-IGS}: Non-robust IGS, which is derived by our proposed scheme for the two-user IC without considering errors in CSI. 
\item {\bf 2-IGS}: The joint variance and covariance IGS algorithm in \cite{zeng2013transmit}.
\item {\bf TS-I}: Robust IGS with time sharing.
\item {\bf TS-P}: Robust PGS with time sharing.
\end{itemize}

\begin{figure}[t]
\centering
\includegraphics{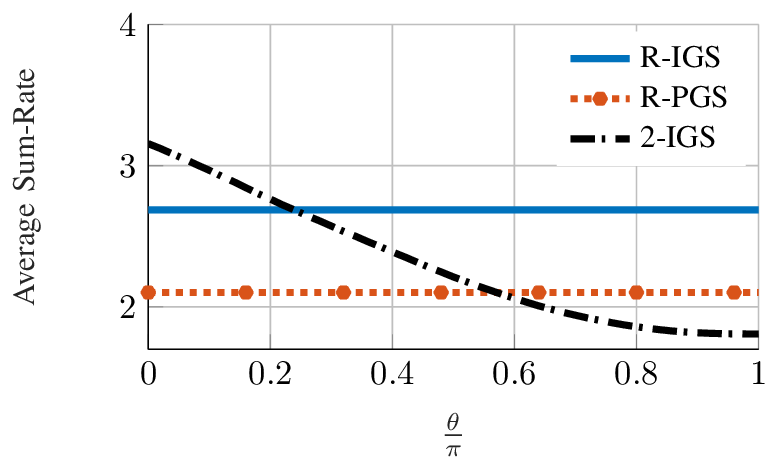}
\caption{Average sum-rate of the two-user IC for 
SNR$=10$dB versus the phase error.}
\label{Fig80}
\end{figure}
\begin{figure*}[t!]
    \centering
    \begin{subfigure}[t]{0.5\textwidth}
        \centering
        \includegraphics{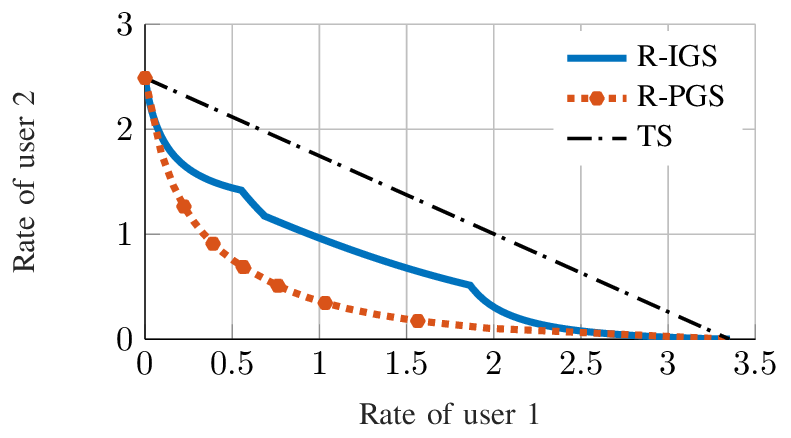}
        \caption{SNR$=10\,$dB.}
    \end{subfigure}%
    ~ 
    \begin{subfigure}[t]{0.5\textwidth}
        \centering
        \includegraphics{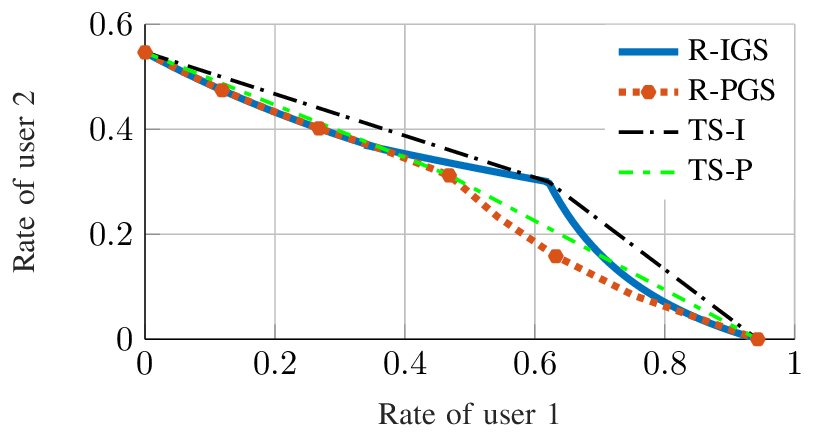}
        \caption{SNR$=0\,$dB.}
    \end{subfigure}
    \caption{Robust rate region of the two-user IC for  $\delta_{ij}=0.5$ for $i,j\in\{1,2\}$ and $i\neq j$.}
	\label{Fig40}
\end{figure*}
\begin{figure*}[t!]
    \centering
    \begin{subfigure}[t]{0.5\textwidth}
        \centering
        \includegraphics[width=\textwidth]{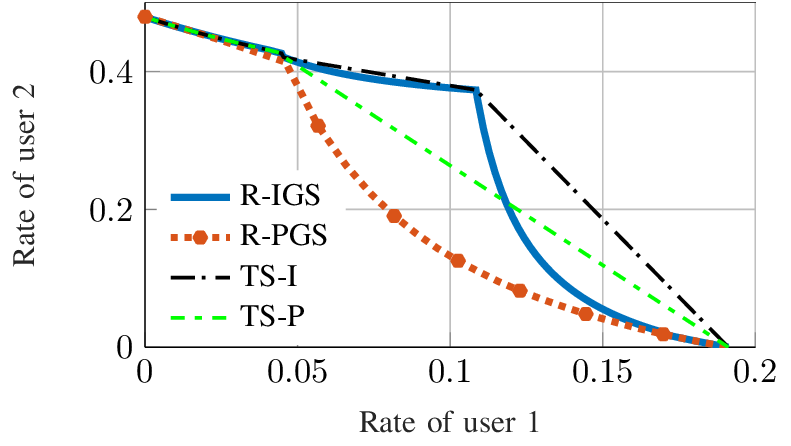}
        \caption{SNR$=10\,$dB.}
    \end{subfigure}%
    ~ 
    \begin{subfigure}[t]{0.5\textwidth}
        \centering
        \includegraphics[width=\textwidth]{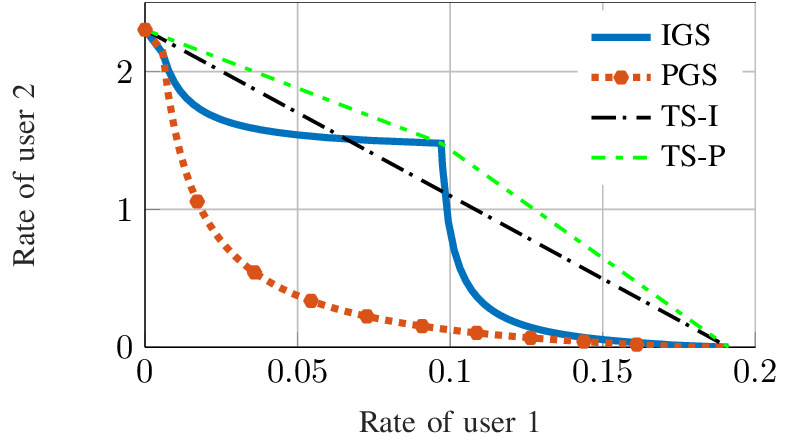}
        \caption{$\text{SNR}_1=0\,$dB, $\text{SNR}_2=10\,$dB.}
    \end{subfigure}
    \caption{Robust rate region for the two-user IC and channel realization $\hat{\mathbf{H}}_2$.}
	\label{Fig4000}
\end{figure*}
In order to evaluate the effect of phase error on the transmitter side, we assume perfect CSI for the channel gains in Fig. \ref{Fig80}.
This figure shows the average sum-rate for 
signal-to-noise-ratio (SNR), i.e., $\frac{P}{\sigma^2}$, equal to $10\,$dB versus the accuracy of the phase information. 
We consider a specific point in the rate region, which is given by the fairness point in the algorithm in \cite{zeng2013transmit}. In order to provide a fair comparison, we first derive the rates by the algorithm in \cite{zeng2013transmit} and then fix the rate of user 1 in our design to the worst-case rate of user 1 achieved by the algorithm in \cite{zeng2013transmit}. We average the results over $200$ channel realizations, and 
consider the worst-case performance of the average sum-rate for a bounded error in the available phase information. 
The considered error is the aggregated error in phase information for both interference and direct links, i.e., $\theta_i=\theta_{ii}+\theta_{ji}$ for $i,j\in\{1,2\}$, where $\theta_{1}=\theta_{2}=\theta$ is the horizontal axis of Fig. \ref{Fig80}. 
The worst-case phase is the phase that minimizes the achievable rate of users. 
As can be observed, our proposed algorithm is suboptimal when perfect CSI is available, compared to the non-robust algorithm proposed in \cite{zeng2013transmit}. 
However, it outperforms this algorithm when the aggregated phase error increases. 
For this example, our robust algorithm performs better than the algorithm in \cite{zeng2013transmit} when $\theta\geq 42^{\circ}$. 
 Our robust proposed algorithm for the two-user IC and robust PGS are both  independent of phase information, hence their performance is independent of phase information. 
 However, the performance of the non-robust algorithm highly depends on the accuracy of the phase information. 
 Thus, our proposed algorithm outperforms existing methods when the error in the phase of the channels is high. 

In Fig. \ref{Fig40}, we show the achievable robust rate region of the two-user IC for SNRs of 10 dB and 0 dB, and estimated channel
\begin{equation}
\hat{\mathbf{H}}_1=\left[ \begin{array}{cc}
1.4833e^{-i2.5864} & 0.6375e^{-i1.3064} \\ 0.9367e^{-i3.0001 }  & 0.4118e^{i2.1234} \end{array} \right],
\end{equation} 
where $[\hat{\mathbf{H}}]_{ij}=\hat{h}_{ij}$. 
In this figure, we consider a scenario with perfect CSI for the direct links but imperfect CSI for the interference links. 
This may represent a scenario where less resources are devoted to acquiring the CSI of the interfering links, which is relevant because interfering signals are treated as noise at the receivers. 
 In this figure, it can be observed that IGS can enlarge the robust rate region in both high and low SNR regimes.
Moreover, the performance improvement of IGS in the high SNR regime is higher than in the low SNR regime for the same channels. For this channel realization, IGS with time sharing (TS)\footnote{The achievable rate region with TS is derived by taking the convex-hull operation over the corresponding achievable rate region of each design \cite{zeng2013transmit, lameiro2017rate}. 
Note that TS yields the convex hull when the power constraint is satisfied in each operating point. The rate region may be enlarged by constraining the average transmit power over the different operation points instead. It is shown in \cite{hellings2017improper, hellings2018improper} for {\em perfect CSI} that IGS with TS does not provide any gain over PGS with TS in the Z-IC and two-user IC. Repeating this analysis for the imperfect CSI case is by no means straightforward and falls outside the scope of the paper.} performs the same as PGS with TS  at $10\,$dB SNR, indicated by TS in  Fig. \ref{Fig40}a. 
However, at $0\,$dB SNR, IGS with TS outperforms PGS with TS in Fig. \ref{Fig40}b.

\begin{figure*}[t!]
    \centering
    \begin{subfigure}[t]{0.5\textwidth}
        \centering
        \includegraphics[width=\textwidth]{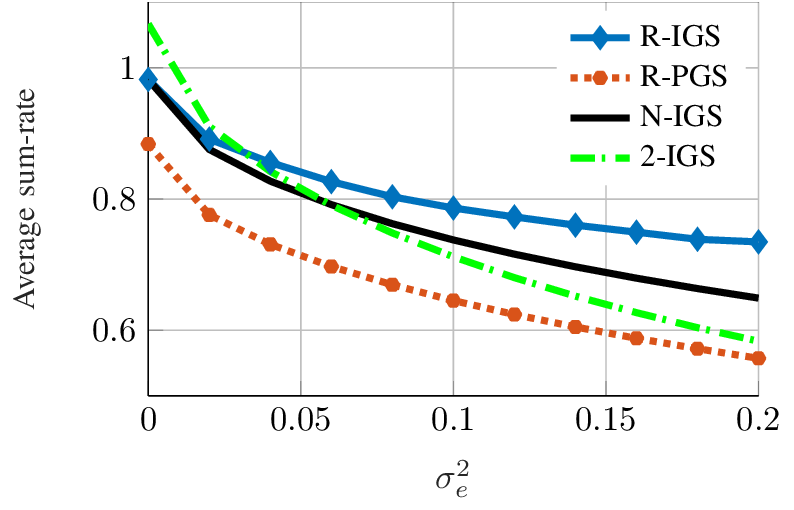}
        \caption{SNR$=0\,$dB.}
    \end{subfigure}%
    ~ 
    \begin{subfigure}[t]{0.5\textwidth}
        \centering
        \includegraphics[width=\textwidth]{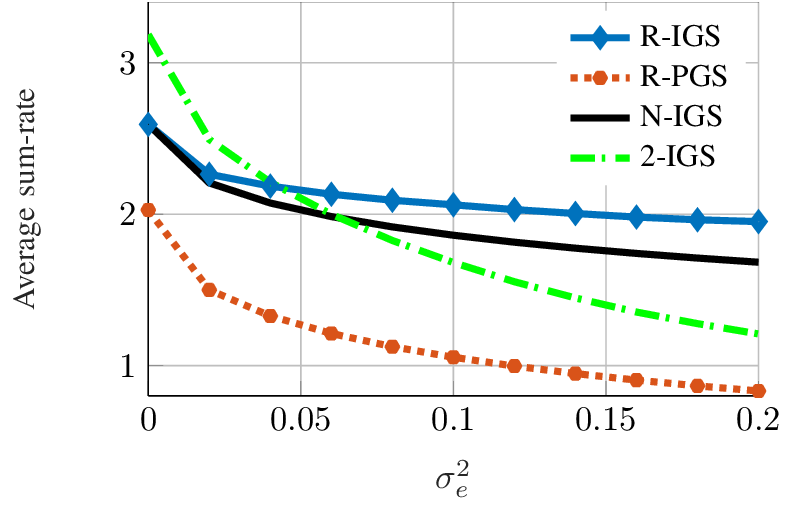}
        \caption{SNR$=10\,$dB.}
    \end{subfigure}
    \caption{Average sum-rate of the two-user IC 
    when only the CSI of the interference links is imperfect.}
	\label{Fig9}
\end{figure*}

In Fig. \ref{Fig4000}, we show the achievable rate region of the two-user IC for perfect CSI, and channel 
 \begin{equation}
\hat{\mathbf{H}}_2=\left[ \begin{array}{cc}
0.3764 e^{i1.4381} & 0.4029e^{i0.9486} \\ 1.8542e^{i2.8153}  & 0.6277e^{i2.3697} \end{array} \right].
\end{equation}
We consider SNR of 10 dB for Fig. \ref{Fig4000}a, and $\text{SNR}_1=0\,$dB, $\text{SNR}_2=10\,$dB for Fig. \ref{Fig4000}b, where $\text{SNR}_i=\frac{P_i}{\sigma^2}$. 
As can be observed, IGS significantly enlarges the achievable rate region for $\hat{\mathbf{H}}_2$. 

In Figs. \ref{Fig9} and \ref{Fig10}, we consider the average sum-rate of the two-user IC for SNRs of $0\,$dB and $10\,$dB. 
In these figures, we consider the same point of the rate region described in Fig. \ref{Fig80} and average the results over 200 channel realizations. We 
set $\Pr=95\%$ (which implies a maximum outage probability of 5\%) and vary the variance of the estimation error, $\sigma^2_e$. 
Moreover, the radius of the uncertainty region is then obtained as \eqref{radius}.

In Fig. \ref{Fig9}, we consider imperfect CSI only for the interference links, where $\delta_{11}=\delta_{22}=0$ and $\delta_{21}=\delta_{12}$. As can be observed in Fig. \ref{Fig9}, the IGS algorithms are always better than PGS.
Since our proposed algorithm for the two-user IC is suboptimal for perfect CSI, the non-robust algorithm in \cite{zeng2013transmit} performs better than our robust design. 
 Specifically, it achieves $18\%$ higher sum-rate for SNR$=10\,$dB and $8\%$ for SNR$=0\,$dB. 
However, as the error in the interference link increases, our proposed algorithm performs better than the non-robust algorithm, where there is a $61\%$ and $25\%$ improvement in achievable sum-rate for $\sigma^2_e=0.2$ when SNR$=10\,$dB and SNR$=0\,$dB, respectively.

\begin{figure*}[t!]
    \centering
    \begin{subfigure}[t]{0.5\textwidth}
        \centering
       \includegraphics[width=\textwidth]{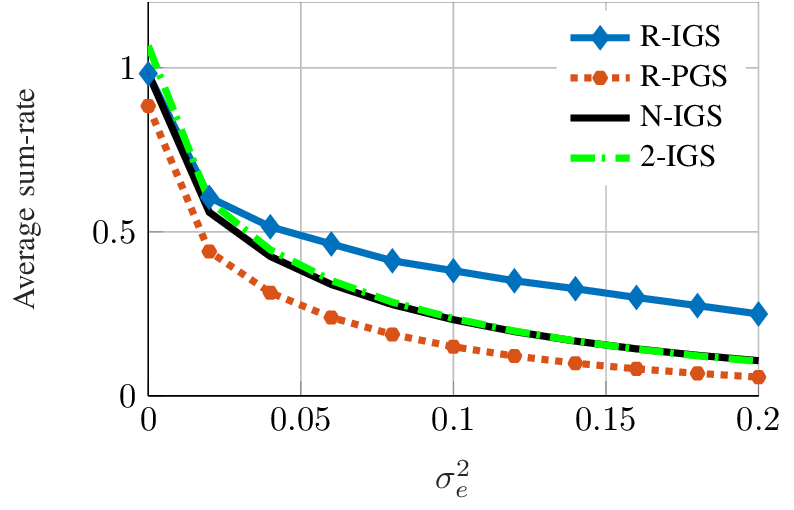}
        \caption{SNR$=0\,$dB.}
    \end{subfigure}%
    ~ 
    \begin{subfigure}[t]{0.5\textwidth}
        \centering
        \includegraphics[width=\textwidth]{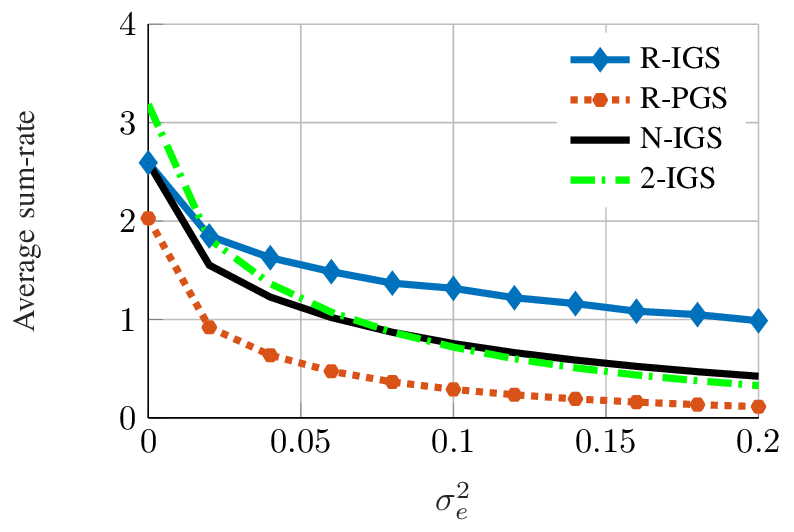}
        \caption{SNR$=10\,$dB.}
    \end{subfigure}
    \caption{Average sum-rate of the two-user IC 
    when the CSI of the direct and interference link is imperfect.}
	\label{Fig10}
\end{figure*}
In Fig. \ref{Fig10}, we consider errors in all links with the same size of the uncertainty region, i.e., $\delta_{ij}=\delta$ for $i,j\in\{1,2\}$. 
We observe that 
the IGS algorithms always perform better than PGS.
As also observed in Fig. \ref{Fig9},
our algorithm performs better than the non-robust algorithm as the error in all links increases. 
For this example, there is a $200\%$ and $137\%$ improvement in average sum-rate when $\sigma^2_e=0.2$ for SNR$=10\,$dB and SNR$=0\,$dB, respectively.

These figures show that our proposed algorithm is robust against imperfect CSI and provides a considerable gain in sum-rate compared to a non-robust approach in high estimation errors. 
This is a rather surprising result, which means that IGS is even more robust to imperfect CSI than its proper counterpart when the transmission parameters are optimized in a robust way.

\subsection{Results for the Z-IC}
\begin{figure}[t]
\centering
\includegraphics{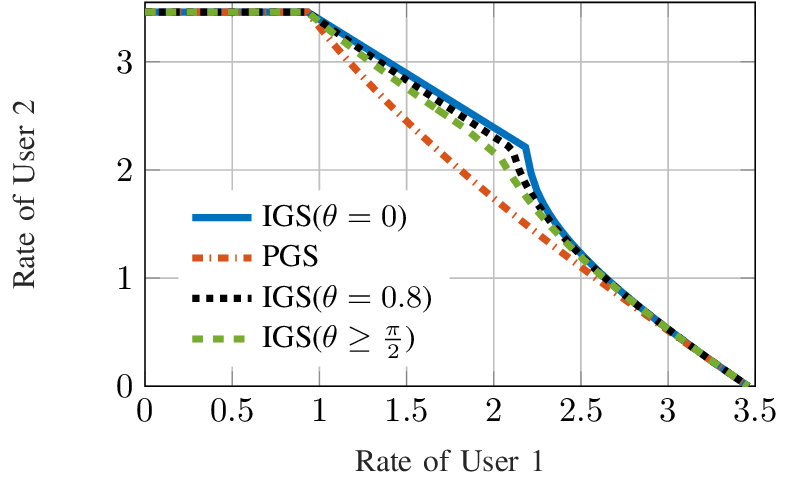}
\caption{Achievable rate region boundaries of the Z-IC for different phase uncertainties and SNR=$10$dB.}
\label{Fig3}
\end{figure}
In this subsection, we evaluate the robust design for the Z-IC. 
In the figures of this subsection, we use the following acronyms:
\begin{itemize}
\item {\bf R-IGS}: Our proposed robust IGS in Section \ref{secIII}.
\item {\bf R-PGS}: Robust PGS.
\item {\bf N-IGS}: The non-robust IGS optimal solution in \cite{lameiro2017rate}.
\item {\bf R-IGS-2IC}: Our proposed robust IGS in Section \ref{IX-ch}.
\end{itemize}
 We first consider the enlarged uncertainty region to illustrate the effect of the aggregated phase uncertainty $\theta$ in the achievable worst-case rate region. 
Figure \ref{Fig3} shows the impact of the aggregated phase uncertainty on the rate region of the Z-IC. In this figure, the channel gains are equal to 1 for all links. 
In order to consider the effect of the phase uncertainty, we ignore at this point the estimation error in the channel gains. 
In other words, the channel gains are perfectly known in this example. 
There is a significant improvement over PGS 
when only user 2 employs IGS (which corresponds to the case $\theta\geq\pi/2$), 
whereas the performance gain slightly increases when user 1 also employs IGS (which corresponds to the case $\theta<\pi/2$). 
Notice that the proposed strategy for $\theta\geq\pi/2$ is 
the same as the one derived for the two-user IC, where only one user employs IGS. 
This result corroborates again that allowing only one user to use IGS is an effective way of improving the performance with imperfect CSI. 
\begin{figure*}[t!]
    \centering
    \begin{subfigure}[t]{0.5\textwidth}
        \centering
        \includegraphics[width=\textwidth]{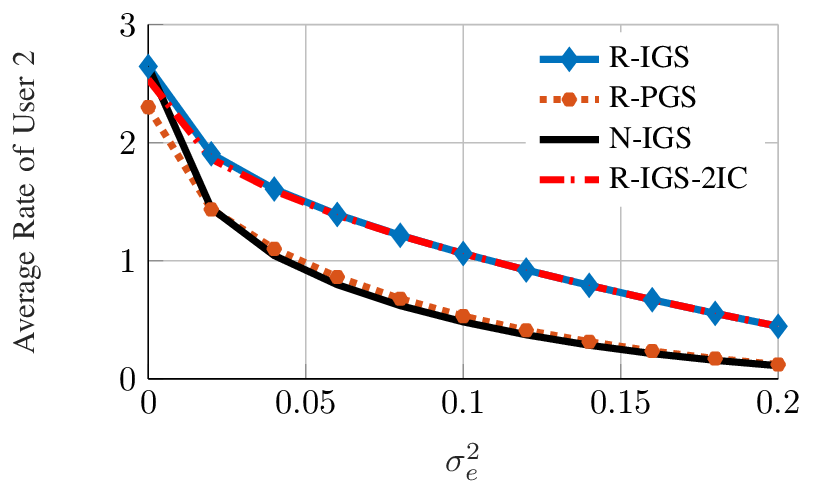}
        \caption{$\alpha=40\%$.}
    \end{subfigure}%
    ~ 
    \begin{subfigure}[t]{0.5\textwidth}
        \centering
        \includegraphics[width=\textwidth]{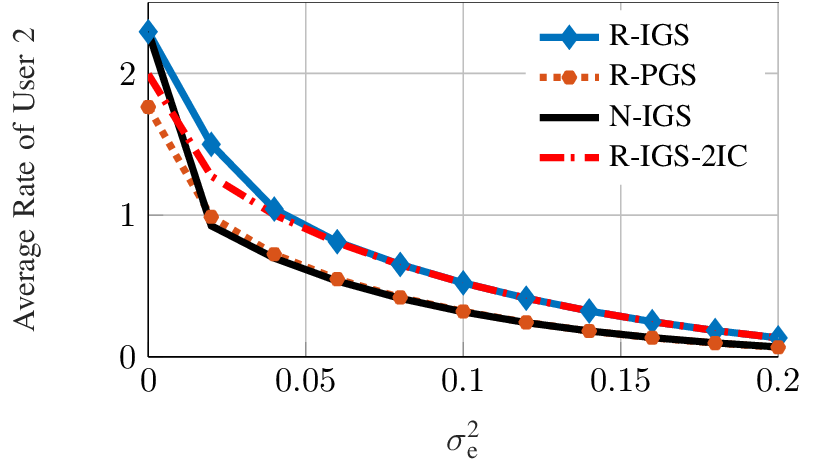}
        \caption{$\alpha=60\%$.}
    \end{subfigure}
    \caption{Average rate of user 2 versus the variance of channel estimation error for SNR=$10$dB in the Z-IC.}
	\label{Fig4}
\end{figure*}
\begin{figure*}[t!]
    \centering
    \begin{subfigure}[t]{0.5\textwidth}
        \centering
        \includegraphics{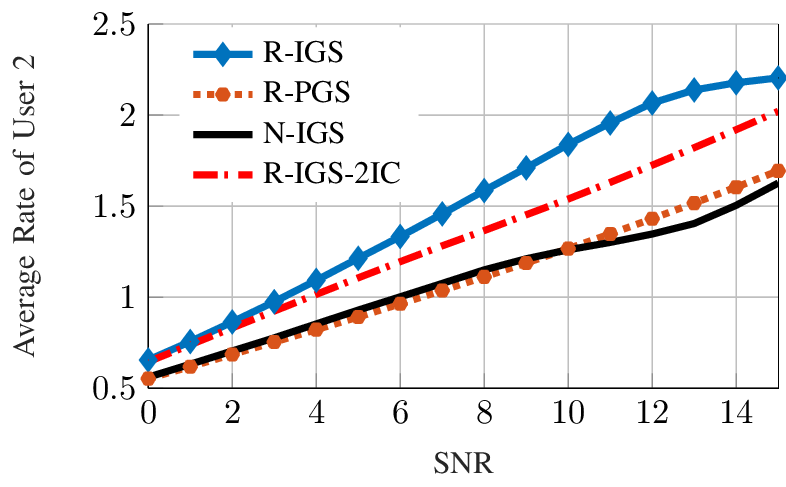}
        \caption{$\sigma^2_e=0.01$.}
    \end{subfigure}%
    ~ 
    \begin{subfigure}[t]{0.5\textwidth}
        \centering
        \includegraphics{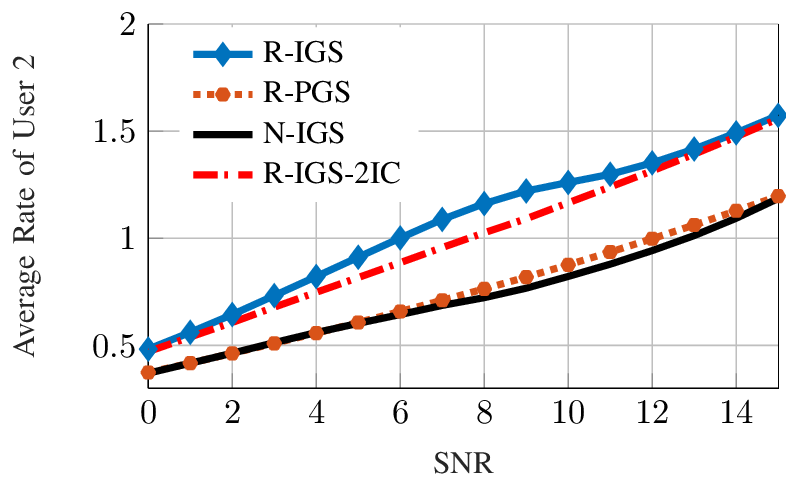}
        \caption{$\sigma^2_e=0.03$.}
    \end{subfigure}
    \caption{Average rate of user 2 versus SNR for $\alpha=60\%$ in the Z-IC.}
	\label{Fig4-n}
\end{figure*}

In Figs. \ref{Fig4} and \ref{Fig4-n}, we average the results over $10^3$ channel realizations.  
In order to provide a fair comparison, we depict the rate of user 2 for a fixed rate of user 1. 
In other words, we reduce the power of user 2 in the non-robust PGS and IGS approaches to achieve the same worst-case rate for user 1 as with the robust schemes. 
The rate of user 1 is fixed to $\alpha\%$  
of its maximum worst-case achievable rate for each channel realization. 

In Fig. \ref{Fig4}, we show the rate of user 2 versus the variance of the estimation error for SNR=$10$dB, $\alpha=40\%$ and $\alpha=60\%$.
Figure \ref{Fig4} shows a gap between PGS and IGS as the channel estimation error increases. 
In Fig. \ref{Fig4-n}, we represent the rate of user 2 versus SNR for $\alpha=60\%$, $\sigma^2_e=0.01$, and $\sigma^2_e=0.03$. As can be observed, the robust IGS significantly outperforms the non-robust IGS as well as robust PGS. 
These figures show that the non-robust IGS does not provide any gain compared to the robust PGS in the Z-IC and may even perform worse than the robust PGS when the CSI is not accurate. 
However, if we design the parameters in a robust way, we can achieve a considerable gain even in presence of highly noisy CSI.
\begin{figure*}
\setcounter{equation}{37}
\begin{align}
R_1^w(p_1,\kappa_2)&=\! \frac{1}{2}\!\log_2\!\left(1+\!\frac{p_1^2|\tilde{h}_{11}|^4\!+2p_1|\tilde{h}_{11}|^2(P_2|\tilde{h}_{12}|^2\!+\!\sigma^2)}{\sigma^4+2\sigma^2P_2|\tilde{h}_{21}|^2+(1-\kappa_2^2)P_2^2|\tilde{h}_{21}|^4}\!\right),\label{eqo-34}\\
\label{eqo-34-new-2cul}
R_2^w(p_1,\kappa_2)&=\! \frac{1}{2}\log_2\left(\frac{P_2^2|\tilde{h}_{22}|^4(1-\kappa_2^2)}{(\sigma^2+p_1|\tilde{h}_{12}|^2)^2}+\frac{2P_2|\tilde{h}_{22}|^2}{\sigma^2+p_1|\tilde{h}_{12}|^2}+1\right).
\end{align}
\setcounter{equation}{41}
\begin{subequations}\label{new-eq-2cul-42}
\begin{alignat}{4}
\frac{\partial\bar{R}_1(p_1)}{\partial p_1}&=\frac{(2\zeta_1p_1+\beta_1)(\zeta_2p_1^2+\beta_2p_1+\tau)-(\zeta_1p_1^2+\beta_1p_1)(2\zeta_2p_1+\beta_2)}{(\zeta_2p_1^2+\beta_2p_1+\tau)^2},\\
&=\frac{(\zeta_1\beta_2-\zeta_2\beta_1)p_1^2+2\zeta_1\tau p_1+\beta_1\tau}{(\zeta_2p_1^2+\beta_2p_1+\tau)^2}.\label{sa-34}
\end{alignat}
\end{subequations}
\hrulefill
\end{figure*}

In these figures, we also compare our robust Z-IC design with our robust design for the two-user IC, in which only one user (user 2) employs IGS. 
As can be observed, our robust Z-IC design performs better than our robust two-user IC design for low channel estimation error; on the other hand, these algorithms perform similarly when the channel estimation error is high. 
This suggests that high-quality CSI permits both users to employ IGS in order to make the signal and interference as close to orthogonal as possible. If the channel phase is not reliable ($\theta\geq\pi/2$), the robust design employs $\kappa_1=0$, which corresponds to the same 
robust design as for the two-user IC. In this case, IGS still provides a significant performance increase (see Fig. \ref{Fig3}), while keeping a  CSI requirement similar to that of PGS.

\section{Conclusion}
In this paper we have studied the robustness of IGS against imperfect CSIT in the single-antenna two-user IC and Z-IC. 
We have proposed robust designs for the two-user IC and Z-IC, which have closed-form solutions for the transmission parameters. 
We have derived closed-form conditions   
 when IGS outperforms PGS for the two-user IC and Z-IC in the presence of imperfect CSI.
We have shown through analytical studies that even if there is no reliable phase information, a robust IGS design can outperform a robust PGS design. 
In this case, one user may employ IGS, while the other employs PGS. 
We evaluated our analytical studies by simulations  
and showed that IGS 
permits a significant performance increase for the two-user IC and Z-IC. 
Our numerical results show that IGS is even more robust to imperfect CSI than PGS provided that the transmission parameters are properly designed, and thus IGS still pays off in the context of imperfect CSI.  


\appendices
\section{Proof of Theorem \ref{th:theorem2}}\label{app:theorem2}
Here, we assume that user 2 employs IGS and transmits with maximum power without loss of generality.  
In this strategy, the worst-case rates of users are given by \eqref{eqo-34} and \eqref{eqo-34-new-2cul} on the top of next page.
If 
$R_2^w(P_1,1)> \alpha R_{2,\text{max}}^w$, maximally improper for user 2 and maximum power transmission for user 1, i.e., $p_1=P_1$ and $\kappa_2=1$, is the solution of \eqref{pareto-2} since it maximizes the rate of user 1. 
Thus, in the following, we consider the case $R_2^w(P_1,1)< \alpha R^w_{2,\text{max}}$. 
The constraint $R_2^w(p_1,\kappa_2)= \alpha R^w_{2,\text{max}}$ is equivalent to
\setcounter{equation}{39}
\begin{equation}\label{eqo-35}
1\!-\!\kappa_2^2=\frac{(\sigma^2+p_1|\tilde{h}_{12}|^2)^2}{P_2^2|\tilde{h}_{22}|^4}\gamma_2(2\alpha)\!-2\frac{\sigma^2\!+\!p_1|\tilde{h}_{12}|^2}{P_2|\tilde{h}_{22}|^2}.
\end{equation}
It can be easily verified that (\ref{eqo-35}) results in (\ref{eqo-305}). 
Moreover, if we consider \eqref{eqo-35} as a quadratic function in $p_1$ and take its positive root, we obtain \eqref{eqo-3050}.
Additionally, plugging (\ref{eqo-35}) into (\ref{eqo-34}) yields \eqref{r-opt-2}.
Through \eqref{eqo-3050}, it is clear that $p_1$ is decreasing in $\kappa_2$. 
Thus, IGS can  improve the performance of the system if and only if (\ref{r-opt-2}) is  decreasing in $p_1$. 
This means that the argument of the logarithm function in (\ref{r-opt-2}) must be  decreasing in $p_1$.
Note that \eqref{r-opt-2} holds if and only if $R^w_2(p_1,\kappa_2)= \alpha R^w_{2,\text{max}}$, which results in $p_1\in [\mathcal{P}(1),\mathcal{P}(0)]$. 
As a result, if $P_1<\mathcal{P}(1)$, \eqref{r-opt-2} does not hold, and maximally IGS is optimal for user 2 as indicated before. 
This case is mentioned as condition 1 of Theorem \ref{th:theorem2}.

In the following, we derive the other conditions. 
Let us define $\bar{R}_1(p_1)$ as
\begin{equation}
\bar{R}_1(p_1)\triangleq\frac{\zeta_1p_1^2+\beta_1p_1}{\zeta_2p_1^2+\beta_2p_1+\tau},
\end{equation}
where $\zeta_1$, $\beta_1$, $\zeta_2$, $\beta_2$, and $\tau$ are defined as in Theorem 1.
The derivative of $\bar{R}_1(p_1)$ with respect to $p_1$ is given by \eqref{new-eq-2cul-42}.
Note that $\zeta_1$, $\zeta_2$ and $\beta_1$ are always positive. Additionally, the denominator of $\frac{\partial\bar{R}_1(p_1)}{\partial p_1}$ is positive. 
Thus, the behavior of $\bar{R}_1(p_1)$ with respect to $p_1$ depends  only on the sign of the numerator of  $\frac{\partial\bar{R}_1(p_1)}{\partial p_1}$, which we denote as $\bar{r}_1(p_1)$.
The sign of $\bar{r}_1(p_1)$ 
is related to $\zeta_1\beta_2-\zeta_2\beta_1$ and $\tau$ as these terms can be either positive or negative. 
If 
$\zeta_1\beta_2-\zeta_2\beta_1$ and $\tau$ have the same sign, $\bar{R}_1(p_1)$ is monotone in $p_1$. 
Otherwise, there is a positive extreme point, which is equal to the positive root of $\bar{r}_1(p_1)$ 
and given by $x_1^{\star}$ in (\ref{eqo-304}) if $(\zeta_1\beta_2-\zeta_2\beta_1)<0$, and $x_2^{\star}$ in \eqref{eqo-304-1} otherwise. 
Thus, $\bar{R}_1(p_1)$ is strictly decreasing in $p_1$ if $\zeta_1\beta_2-\zeta_2\beta_1<0$ and $\tau<0$. Hence, the rate of user 1 is maximized if user 1 transmits with $p_1=\mathcal{P}(1)$, and user 2 employs maximally IGS, which results in condition 2 of Theorem \ref{th:theorem2}. 
If $\zeta_1\beta_2-\zeta_2\beta_1>0$ and $\tau>0$, $\bar{R}_1(p_1)$ is strictly increasing in $p_1$, and PGS is the optimal solution for both users.
\begin{figure*}[t!]
    \centering
    \begin{subfigure}[t]{0.5\textwidth}
        \centering
        \includegraphics[width=\textwidth]{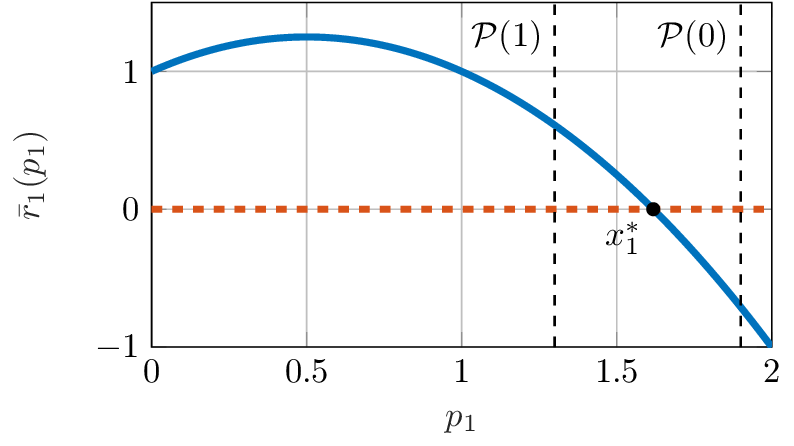}
        \caption{$\zeta_1\beta_2-\zeta_2\beta_1=-\zeta_1\tau p_1=-\beta_1\tau=-1$. }
    \end{subfigure}%
    ~ 
    \begin{subfigure}[t]{0.5\textwidth}
        \centering
        \includegraphics[width=\textwidth]{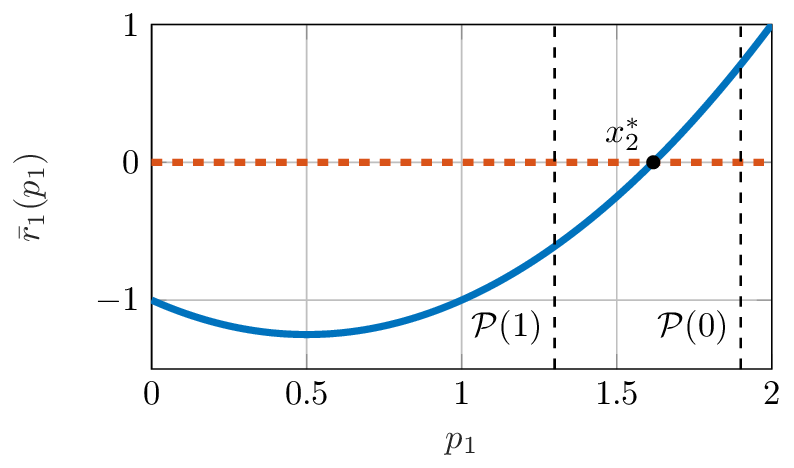}
        \caption{$\zeta_1\beta_2-\zeta_2\beta_1=-\zeta_1\tau p_1=-\beta_1\tau=1$.}
    \end{subfigure}
    \caption{Function $\bar{r}_1(p_1)$ versus $p_1$.}
	\label{R1-3}
\end{figure*}

Now we consider the case that $\bar{R}_1(p_1)$ is not  necessarily monotone in $p_1\in [\mathcal{P}(1),\mathcal{P}(0)]$, i.e., when $(\zeta_1\beta_2-\zeta_2\beta_1)\tau<0$.
Figure \ref{R1-3} 
shows the behavior of $\bar{r}_1(p_1)$ when $(\zeta_1\beta_2-\zeta_2\beta_1)\tau<0$. 
Note that the extreme point can be outside of $[\mathcal{P}(1),\mathcal{P}(0)]$. 
In other words, there is at most one extreme point in $p_1\in [\mathcal{P}(1),\mathcal{P}(0)]$. 
Since $\bar{r}_1(p_1)$ is a quadratic function of $p_1$, $\bar{r}_1(p_1)$ is either convex or concave. 
Let us first consider the concave case, i.e., $\zeta_1\beta_2-\zeta_2\beta_1<0$ and $\tau>0$. 
In this case, if $\mathcal{P}(0)<x_1^{\star}$, $\bar{R}_1(p_1)$ is strictly increasing in $p_1\in [\mathcal{P}(1),\mathcal{P}(0)]$, and thus, $p_1=\mathcal{P}(0)$ and PGS is optimal for both users. 
Therefore, we consider $\mathcal{P}(0)>x_1^{\star}$, which implies $\frac{\partial R^w_1(p_1)}{\partial p_1}|_{p_1=\mathcal{P}(0)}<0$. 
In this case, $\bar{R}_1(p_1)$ is maximized at $p_1=x^{\star}_1$ if $x_1^{\star}\in [\mathcal{P}(1),\mathcal{P}(0)]$, and at $p_1=\mathcal{P}(1)$ otherwise, which results in condition 3 of Theorem \ref{th:theorem2}.

Now we consider the case that $\bar{r}_1(p_1)$ is convex, i.e., $\zeta_1\beta_2-\zeta_2\beta_1>0$ and $\tau<0$. In this case, $\bar{R}_1(p_1)$ is strictly increasing in $p_1\in [\mathcal{P}(1),\mathcal{P}(0)]$ if $\mathcal{P}(1)>x_2^{\star}$, where $x^{\star}_2$ is given by \eqref{eqo-304}, which results in PGS being optimal. 
Moreover, $\bar{R}_1(p_1)$ is strictly decreasing in $p_1\in [\mathcal{P}(1),\mathcal{P}(0)]$ if $\mathcal{P}(0)<x_2^{\star}$, and thus, $p_1=\mathcal{P}(1)$ and maximally IGS is optimal for user 2 in this case.
If $x_2^{\star}\in [\mathcal{P}(1),\mathcal{P}(0)]$, $R_1(p_1)$ is strictly increasing in $p_1\in [x^{\star}_2,\mathcal{P}(0)]$ and strictly decreasing in $p_1 \in [\mathcal{P}(1),x_2^{\star}]$, as illustrated in Fig. \ref{R1-3}b.
Thus, IGS is beneficial in this case only if $R^w_1(\mathcal{P}(1))>R^w_1(\mathcal{P}(0))$, and maximally IGS is then the optimal strategy, which results in condition 4 of Theorem \ref{th:theorem2}.

\section{Proof of Theorem \ref{th:theorem1}}\label{app:theorem1}
\begin{figure*}
\hrulefill
\setcounter{equation}{43}
\begin{equation}\label{20}
2\frac{\partial q(\kappa_2,\theta)}{\partial \kappa^2_2}\!=\!\frac{q^2(\kappa_2,\theta)|\tilde{h}_{21}|^2(\gamma_1(2\alpha)+\cos^2\theta)}{q(\kappa_2,\theta)|\tilde{h}_{21}|^2\left[(\gamma_1(2\alpha)+1)(1-\kappa^2_2)-1+\kappa^2_2\sin^2\theta\right]+ \sigma^2\gamma_1(2\alpha)-P_1|\tilde{h}_{11}|^2}.
\end{equation}
\begin{equation}\label{21}
\frac{|\tilde{h}_{21}|^2(\gamma_1(2\alpha)+\cos^2\theta)\left[|\tilde{h}_{22}|^2q(\kappa_2,\theta)\sigma^{-2}(1-\kappa^2_2)+1\right]}{q(\kappa_2,\theta)|\tilde{h}_{21}|^2\left[(\gamma_1(2\alpha)+1)(1-\kappa^2_2)-1+\kappa^2_2\sin^2\theta\right]+\sigma^2\gamma_1(2\alpha)-P_1|\tilde{h}_{11}|^2}> |\tilde{h}_{22}|^2\sigma^{-2}.
\end{equation}
\setcounter{equation}{46}
\begin{equation}\label{eq-new-2000}
2\frac{\partial q(\kappa_2,\theta)}{\partial \kappa^2_2}=\frac{q(\kappa_2,\theta)|\tilde{h}_{21}|^2\gamma_1(2\alpha)-\frac{q(\kappa_2,\theta)P_1|\tilde{h}_{11}|^2\cos\theta}{\kappa_2}}{q(\kappa_2,\theta)|\tilde{h}_{21}|^2\gamma_1(2\alpha)(1-\kappa^2_2)+ \sigma^2\gamma_1(2\alpha)-P_1|\tilde{h}_{11}|^2(1+\kappa_2\cos\theta)}.
\end{equation}
\begin{equation}\label{tt19}
\frac{[q(\kappa_2,\theta)|\tilde{h}_{21}|^2\gamma_1(2\alpha)-\frac{P_1|\tilde{h}_{11}|^2\cos\theta}{\kappa_2}]\left[|\tilde{h}_{22}|^2q(\kappa_2,\theta)(1-\kappa^2_2)+\sigma^{2}\right]}{q(\kappa_2,\theta)|\tilde{h}_{21}|^2\gamma_1(2\alpha)(1-\kappa^2_2)+ \sigma^2\gamma_1(2\alpha)-P_1|\tilde{h}_{11}|^2(1+\kappa_2\cos\theta)}> q(\kappa_2,\theta)|\tilde{h}_{22}|^2.
\end{equation}
\begin{equation}\label{soli}
q(\kappa_2,\theta)\left[\!(|\tilde{h}_{21}|^2\!\!-\!|\tilde{h}_{22}|^2)\gamma_1(2\alpha)\sigma^2\!+\!P_1|\tilde{h}_{11}|^2|\tilde{h}_{22}|^2\!\left(\!1\!+\!2\kappa_2\cos\theta\!-\!\frac{\cos\theta}{\kappa_2}\!\right)\!\right]\!-\!\sigma^2P_1|\tilde{h}_{11}|^2\frac{\cos\theta}{\kappa_2}>0.
\end{equation}
\hrulefill
\end{figure*}

We first derive the condition that results in $\left.\frac{\partial R^w_2(\kappa_2,\theta)}{\partial \kappa_2^2}\right|_{\kappa_2=0}>0$. Then, we prove that if $\left.\frac{\partial R^w_2(\kappa_2,\theta)}{\partial \kappa^2_2}\right|_{\kappa_2=0}>0$ holds, we have $\frac{\partial R^w_2(\kappa_2,\theta)}{\partial \kappa^2_2}>0$ for $\kappa_2>0$. The term $\frac{\partial R^w_2(\kappa_2,\theta)}{\partial \kappa^2_2}>0$ is equivalent to
\setcounter{equation}{42}
\begin{multline}\label{19}
2\frac{\partial q(\kappa_2,\theta)}{\partial \kappa^2_2}\left[|\tilde{h}_{22}|^2q(\kappa_2,\theta)\sigma^{-2}(1-\kappa^2_2)+1\right]\\
> q^2(\kappa_2,\theta)|\tilde{h}_{22}|^2\sigma^{-2}.
\end{multline}
Now we first consider (\ref{q22}) for $\kappa_1<1$, and then prove that if $\frac{\partial R^w_2(\kappa_2,\theta)}{\partial \kappa^2_2}|_{\kappa_2=0}>0$, the rate of user 2 remains strictly increasing in $\kappa_2$ when $q(\kappa_2,\theta)$ is derived based on the expression in (\ref{q22}) for $\kappa_1=1$. We can obtain $\frac{\partial q(\kappa_2,\theta)}{\partial \kappa^2_2}$ by taking the derivative of (\ref{q22}) with respect to $\kappa_2^2$ for $\kappa_1<1$. That is given by \eqref{20} on the top of next page.
By replacing (\ref{20}) in (\ref{19}), we have \eqref{21}. 
Equation (\ref{21}) is simplified to
\setcounter{equation}{45}
\begin{equation}\label{22}
\frac{|\tilde{h}_{21}|^2}{|\tilde{h}_{22}|^2}>\frac{\sigma^2\gamma_1(2\alpha)-P_1|\tilde{h}_{11}|^2-q(\kappa_2,\theta)|\tilde{h}_{21}|^2\cos^2\theta}{\sigma^{2}(\gamma(2\alpha)+\cos^2\theta)}.
\end{equation}
From (\ref{q2}), it can be seen that $q(\kappa_2,\theta)|\tilde{h}_{21}|^2$ is independent of $\frac{|\tilde{h}_{21}|^2}{|\tilde{h}_{22}|^2}$, and $q_2(\kappa_2=0,\theta)|\tilde{h}_{21}|^2=\frac{P_1}{\gamma_1(\alpha)}-\sigma^2$ which results in (\ref{k2}). 
Moreover, $q(\kappa_2,\theta)$ is a strictly increasing function of $\kappa_2$, which implies that the inequality holds when $\kappa_2>0$. 
Now we prove that if (\ref{22}) holds, $\frac{\partial R^w_2(\kappa_2,\theta)}{\partial \kappa^2_2}>0$ for the case $\kappa_1=1$.
In this case, we have \eqref{eq-new-2000}.
Thus, $\left.\frac{\partial R^w_2(\kappa_2,\theta)}{\partial \kappa^2_2}\right|_{\kappa_2=0}>0$ is equivalent to \eqref{tt19}.
The expression in (\ref{tt19}) is simplified to \eqref{soli}.
It is easy to see that the left-hand side of (\ref{soli}) is strictly increasing in $\kappa_2$ as we are considering $0\leq \theta\leq\frac{\pi}{2}$. Therefore, if (\ref{soli}) holds for $\kappa_2=\overline{\kappa}$, it will hold as well for $\kappa_2>\overline{\kappa}$. 
We now show that such a $\overline{\kappa}$ exists. We have proved that, if (\ref{22}) holds, we have $\frac{\partial R^w_2(\kappa_2,\theta)}{\partial \kappa^2_2}>0$ as long as $q(\kappa_2,\theta)$ is derived from the first equation in (\ref{q22}). 
This implies that $\frac{\partial R^w_2(\kappa_2,\theta)}{\partial \kappa^2_2}>0$ for $\kappa_2=\overline{\kappa}=\frac{P_1|\tilde{h}_{11}|^2}{q(\kappa_2,\theta)|\tilde{h}_{21}|^2\cos\theta}$, where $\overline{\kappa}$ is the smallest value of $\kappa_2$ such that $\kappa_1=1$ (see (\ref{k1})). Hence, \eqref{21} and \eqref{22} are equivalent for $\kappa_2=\overline{\kappa}$, which results in (\ref{soli}) holding as well for $\kappa_2=\overline{\kappa}$. As a result, the rate of user 2 is strictly increasing in $\kappa_2$ if $\frac{\partial R^w_2(\kappa_2,\theta)}{\partial \kappa^2_2}>0$ holds for $\kappa_2=0$.
Thus, we choose the maximum possible $\kappa_2$ when improper signaling is optimal, which results in (\ref{optk22}).

\bibliographystyle{IEEEtran}
\bibliography{ref2}
\end{document}